\documentclass{llncs}
\usepackage{pgf}
\usepackage{tikz}
\usetikzlibrary{arrows,automata}
\usepackage{amsmath}
\usepackage{todonotes}
\usepackage[utf8]{inputenc}
\usepackage{url}
\usepackage{alltt}
\usepackage{amssymb}
\usepackage[english]{babel}
\usepackage{proof}
\usepackage{framed}

\usepackage{tikz}
\usetikzlibrary{trees}
\usetikzlibrary{arrows}
\usetikzlibrary{decorations.pathmorphing}
\usetikzlibrary{shapes.multipart}
\usetikzlibrary{shapes.geometric}
\usetikzlibrary{calc}
\usetikzlibrary{positioning} 
\usetikzlibrary{fit}
\usetikzlibrary{backgrounds}

\usepackage[colorlinks,hyperindex,bookmarks,linkcolor=blue,citecolor=blue,urlcolor=blue]{hyperref}

\renewenvironment{proof}{\emph{Proof.}}{\hspace*{\fill}{\mbox{$\Box$}}}

\newcommand{\secref}[1]{Sect.~\ref{#1}}

\newcommand{\figref}[1]{Fig.~\ref{#1}}

\newcommand{\theoremref}[1]{Theorem~\ref{#1}}

\newcommand{\lemmaref}[1]{Lemma~\ref{#1}}

\newcommand{\defref}[1]{Def~\ref{#1}}

\newcommand{\union}{\cup}

\newcommand{\verule}{\rule{1pt}{10pt}\kern 0.2em}

\newcommand{\remms}[2][]{}

\begin{document}

\title{The satisfiability problem of SROIQ$^\sigma$ is decidable}

\author{ 
Jon Ha\"el Brenas$^1\;\;$ Rachid Echahed$^1\;\;$Martin Strecker$^2$\\
}
\institute{$^1$ CNRS and University of Grenoble\\$^2$ Université de Toulouse / IRIT}

\maketitle

\begin{abstract}
We consider a dynamic extension of the description logic
$\mathcal{SROIQ}$. This means that interpretations could evolve thanks to some
actions such as addition and/or deletion of an element (respectively, a pair
of elements) of a concept (respectively, of a role).  The obtained logic is
called $\mathcal{SROIQ}$ with explicit substitutions and is written
$\mathcal{SROIQ^\sigma}$. Substitution is not treated as meta-operation that
is carried out immediately, but the operation of substitution may be delayed,
so that sub-formulae of
$\mathcal{SROIQ}^\sigma$ are of the form $\Phi\sigma$, where $\Phi$ is a
$\mathcal{SROIQ}$ formula and $\sigma$ is a substitution which encodes changes
of concepts and roles.  In this paper, we particularly prove that the
satisfiability problem of $\mathcal{SROIQ}^\sigma$ is decidable.
\end{abstract}

\section{Introduction}\label{sec:intro}

Description Logics \cite{Baader2003} are logical formalisms for representing
information about classes and objects. They are very often used as the basis
of knowledge representation systems and have been used recently to develop OWL
semantic web language, which can be viewed as an expressive Description Logic
(DL).

There is an impressive variety of DLs. What they all have in
common is that they are tailored to describe an established and fixed ontology
and reason about its properties, even if they differ with respect to
expressiveness and complexity of reasoning primitives.

Our goal here is to consider properties over dynamic ontologies. We introduce
the notion of substitutions to express changes of ontologies (e.g. addition or
deletion of an element, respectively a pair of elements, of a concept,
respectively of a role). We investigate the addition of such substitutions in
the particular case of the logic $\mathcal{SROIQ}^\sigma$, an extension of the
logic $\mathcal{SROIQ}$ \cite{horrocks_kutz_sattler_kr_2006}. We mainly show
that the problem of satisfiability in $\mathcal{SROIQ}^\sigma$ is still
decidable.

Our interest in the topic has arisen out of previous work on verification of
transformations of graph structures
\cite{brenas14:_hoare_like_calcul_using_sroiq,chaabani:dl_transfo2013}. These
structures can be understood as models of DL formulae. We describe
transformations by imperative programs with an appropriately adapted
instruction set. The substitutions of $\mathcal{SROIQ}^\sigma$ arise when
computing weakest preconditions. More in general and apart from this
particular application context, our results can be grafted on other
formalisms, such as the variant of the DL $\mathcal{ALC}$ considered in
\cite{ahmetaj14:_manag_chang_graph_data_using_descr_logic}.

The paper is structured as follows: \secref{sec:syntax} introduces $\mathcal{SROIQ}^\sigma$, our version of $\mathcal{SROIQ}$ with substitutions. \secref{sec:interp} defines the interpretations for that logic. \secref{sec:decidability} proves that $\mathcal{SROIQ}^\sigma$ is decidable by transforming formulae in $\mathcal{SROIQ}^\sigma$ to formulae in $\mathcal{SROIQ}$. \secref{sec:example} illustrates such transformations. Finally, \secref{sec:conclusion} concludes the paper.

\section{Syntax}\label{sec:syntax}

In this section, we define the syntax of the logics $\mathcal{SROIQ}$
and $\mathcal{SROIQ^\sigma}$. We start by introducing
$\mathcal{SROIQ}$\cite{horrocks_kutz_sattler_kr_2006}.

Usually, Description Logics are split into Aboxes (containing
assertions like \textbf{C}(\textsf{a}), meaning that element
\textsf{a} satisfies concept \textbf{C}), Tboxes (containing concept
inclusions like \textbf{C} $\subseteq$ \textbf{D}, meaning that each
element satisfying \textbf{C} also satisfies \textbf{D}) and Rboxes
(containing role axioms). It is proved in \cite{horrocks_kutz_sattler_kr_2006} that
Aboxes and Tboxes can be reduced. Thus we focus on $\mathcal{SROIQ}$
concepts and role axioms only.

\begin{definition}[Signature, Concept name; Nominals; Role names;
  Individuals; Signature]
Let \textbf{C} be a set of \textbf{concept names} and \textbf{N}  a set
of \textbf{nominals} with \textbf{C} and \textbf{N} being disjoint, \textbf{R} a set of \textbf{role
names} including the universal role $U$ and \textbf{I} a set of
\textbf{individuals}. The set of roles is \textbf{R} $\union \{R^- | R
\in \textbf{R}\}$, where a role $R^-$ is called the \textbf{inverse
role} of $R$. We define the signature $\Sigma$ as $\Sigma = (\textbf{C},\textbf{N},\textbf{R},\textbf{I})$.
\end{definition}

We will use a running example, the description of a family, as a way to illustrate the diverse
components of our logic. 

\begin{example}
 In this example, \textbf{C} will be \{\textbf{Animal}, \textbf{Female}, \textbf{Male}
\}, \textbf{N} will be \{\texttt{Alice}, \texttt{Bob},
\texttt{Charles}\}, \textbf{R} will be \{$U$, \textsc{Offspring}, \textsc{Parent},
\textsc{Owner}, \textsc{Brother}, \textsc{Sister},
\textsc{FamilyMember}\} and \textbf{I} will be \{\textsf{Alice}\}.
\end{example}


In $\mathcal{SROIQ}$, one can provide role axioms that state global
properties of roles. They can take two different forms: role
hierarchies and role assertions. For decidability reasons, we need the notions of
regular hierarchy over roles as well as simple roles. We define these notions hereafter.

\begin{definition}[Regular order]
A strict partial order $\prec$ on a set $A$ is an irreflexive and transitive relation on $A$. A strict partial order $\prec$ on the set of roles is called a \textbf{regular order} if $\prec$ satisfies, additionally, $S \prec R \Leftrightarrow S^- \prec R$ for all roles $R$ and $S$. 
\end{definition}
\begin{definition}[Role inclusion axiom]
  A \textbf{role inclusion axiom} is an expression of the form $w \subseteq R$
  where \textit{w} is a finite string of roles not containing the universal role \textit{U} and
  \textit{R} is a role name, with \textit{R} $\neq$\textit{U}  . A \textbf{role hierarchy} $\mathcal{R}_h$ is a
  finite set of role inclusion axioms.  A role inclusion axiom $w
  \subseteq R$ is $\prec$\textbf{-regular} if \textit{R} is a role name and \textit{w}
  is defined by the following grammar: \\
\textit{w} = \textit{RR} \verule  $R^-$ \verule $S_1\dots S_n$ \verule $RS_1\dots S_n$ \verule $S_1\dots S_nR$ with $S_i\prec$ R for all 1 $\leq i \leq n$.   Finally, a role hierarchy $\mathcal{R}_h$ is \textbf{regular} if there
  exists a regular order $\prec$ such that each role inclusion axiom in
  $\mathcal{R}_h$ is $\prec$-regular.
\end{definition}

\begin{example}
For instance, \textsc{Brother} $\subseteq$ \textsc{FamilyMember} (that
is the brother of a person is part of her family) and
\textsc{Father} \textsc{Brother} $\subseteq$ \textsc{Father} (that is
the father of a person's brother is her father) are role inclusion
axioms that make sense.
\end{example}

The second possible kind of role axiom is the role assertion.

\begin{definition}[Role assertion]
For roles $R$,$S$, we call the assertions $Ref(R)$, $Irr(R)$,
$Sym(R)$, $Asy(R)$, $Tra(R)$ and $Dis(R,S)$ \textbf{role
assertions}. They, respectively, mean that $R$ is reflexive,
irreflexive, symmetric, asymmetric, transitive and that $R$ and $S$
are disjoint.
\end{definition}

\begin{example}
For instance, $Tra$(\textsc{FamilyMember}) and $Irr$(\textsc{Father})
are possible role assertions.
\end{example}

\begin{definition}[Simple role; Simple assertion]
  Given a role hierarchy $\mathcal{R}_h$ and a set of role assertions
  $\mathcal{R}_a$, a simple role is inductively defined as either a role name
  that does not occur in the right-hand side of any role inclusion axiom, or $R^-$
  for $R$ simple, or the right-hand side of a role inclusion axiom  $w \subseteq
  R$ where $w$ is a simple role.  $\mathcal{R}_a$ is called \textbf{simple} if
  all roles appearing in role assertions are simple.
\end{definition}

From now on, the only role hierarchies that we consider are
regular and the only sets of role assertions that we consider are
finite and simple.


The definition of the concept constructors is the difference between
$\mathcal{SROIQ}$ and $\mathcal{SROIQ}^\sigma$. Below is the
definition of the concepts of the logic $\mathcal{SROIQ}$.

\begin{definition}[$\mathcal{SROIQ}$ concepts]\label{SROIQConcept}
The set of concepts is defined as the smallest set containing:\\
\begin{tabular}[h]{lcll}
   $C$ & ::= & $\bot$     &         \mbox{(empty concept)} \\
       & $|$ & $c$        &         \mbox{(concept name)} \\
       & $|$ & $\neg\; C$     &    \mbox{(negation)} \\
       & $|$ & $C\; \sqcap\; D$   &   \mbox{(conjunction)} \\
       & $|$ & $C\; \sqcup\; D$  &    \mbox{(disjunction)} \\
       & $|$ & $(\geq n\; S\; C)$   &     \mbox{(at least)} \\
       & $|$ & $(< n\; S\; C)$   &     \mbox{(no more than)} \\
       & $|$ & $\exists R.C$ & \mbox{(exists)}\\
       & $|$ & $\forall R.C$ & \mbox{(for all)}\\
       & $|$ & $ o $ &  \mbox{(nominal)} \\
       & $|$ & $\exists R.Self$ & \mbox{(local reflexivity)}
\end{tabular}

where $c$ is a concept name, $R$ is a role, $S$ is a simple role, $o$ is a
nominal and $C$, $D$ are concepts.
\end{definition}

\begin{example}
For example, the concept $\mathfrak{C}$ = (\texttt{Alice} $\sqcap\; (\exists$ 
\textsc{Brother}.\texttt{Bob}) $\sqcap\; \forall$ \textsc{Sister}$^-$.\textbf{Male})
$\sqcup$ ($(< 3\; \textsc{Parent}\; \top) \sqcap (\geq 1\; \textsc{Owner}^-\;
\textbf{Animal}) \sqcap \exists \textsc{FamilyMember}.Self$) is the
concept satisfied by Alice (\texttt{Alice}) if she has a brother named
Bob ($\exists$ 
\textsc{Brother}.\texttt{Bob}) and if all the persons whom she is a sister
of are males ($\forall$ \textsc{Sister}$^-$.\textbf{Male}) or by
anyone having strictly less than 3 parents ($(< 3\; \textsc{Parent}\;
\top)$), who is the owner of at least 1 animal ($(\geq 1\; \textsc{Owner}^-\;
\textbf{Animal})$) and who is a member of her own family ($\exists \textsc{FamilyMember}.Self$).
\end{example}

In order to define the concepts of $\mathcal{SROIQ}^\sigma$, we
need to introduce the notion of substitution. Substitutions, given in the definition below, are intended to modify
concepts (respectively roles) by adding or removing individuals
(respectively pairs of individuals).

\begin{definition}[Substitution]
Given a role name $R$, a concept name $C$ and individuals $i$ and $j$,
a substitution $subst$ is: \\
\begin{tabular}[h]{lcll}
   $subst$ & ::= & $\epsilon$  & \mbox{(empty substitution)}       \\
       & $|$ & $[RS]$ & \mbox{(role substitution)}        \\
       & $|$ & $[CS]$ & \mbox{(concept substitution)}
\end{tabular}\\
A role substitution can be either:\\
 \begin{tabular}[h]{lcll}
   $RS$ & ::= &   $R := R - (i, j)$ &         \mbox{(deletion of relation instance)} \\
       & $|$ &  $R := R + (i, j)$&         \mbox{(insertion of relation instance)} \\
\end{tabular}\\
while a concept substitution is either:\\
\begin{tabular}[h]{lcll}
   $CS$ & ::= &  $c := c - i$ &         \mbox{(deletion of a concept instance)} \\
       & $|$ &  $c := c + i$&         \mbox{(insertion of a concept instance)} \\
\end{tabular}\\
\end{definition}

\begin{definition}[$\mathcal{SROIQ}^\sigma$ concepts]
The concepts of the logic $\mathcal{SROIQ}^\sigma$ are those given
in \defref{SROIQConcept} in addition to the following explicit
substitution constructor: 
\begin{tabular}[h]{lcll}
   $C$ & ::= & ($\mathcal{SROIQ}$ concept) &  \mbox{(\defref{SROIQConcept})}\\
       & $|$ & $C subst $ & \mbox{(explicit substitution)}
\end{tabular}
\end{definition}

\begin{example}
Roughly speaking, the $\mathcal{SROIQ}^\sigma$ concept ($\exists
\textsc{Sister}.\textbf{Female}$)[\textbf{Female} := \textbf{Female} +
\textsf{Alice}] expresses the fact
``there exists a sister who is female'' \emph{once the concept name Female has
been dynamically modified to include Alice}. 
\end{example}


\section{Interpretations and models}\label{sec:interp}

\begin{definition}\label{def:defintnames}
Let $\Sigma = (\textbf{C}, \textbf{N}, \textbf{R}, \textbf{I})$ be a
signature. As usual, an \textbf{interpretation} (over $\Sigma$) $\mathcal{I} =
(\Delta^{\mathcal{I}},.^{\mathcal{I}})$ consists of a set
$\Delta^{\mathcal{I}}$, called the \textbf{domain} of $\mathcal{I}$,
and a \textbf{valuation} $.^{\mathcal{I}}$ which associates to every
concept name $C$ a set $C^{\mathcal{I}}$ such that $C^{\mathcal{I}}
\subseteq \Delta^{\mathcal{I}}$ and to each role name $R$ a binary
relation $R^{\mathcal{I}}$ such that $R^{\mathcal{I}} \subseteq
\Delta^{\mathcal{I}} \times \Delta^{\mathcal{I}}$. The valuation of a
nominal $o$ is a singleton and the valuation of the
universal role $U$ is the universal relation $\Delta^{\mathcal{I}}
\times \Delta^{\mathcal{I}}$. Finally, $\forall x \in \textbf{I}$,
$x^\mathcal{I} \in \Delta^{\mathcal{I}}$. 
\end{definition}

As usual, $(R^-)^{\mathcal{I}} = \{(y,x) | (x,y) \in
R^{\mathcal{I}}\}$ and for $w = Rw_1$, $w^\mathcal{I} = \{(x,y) \vrule \exists z
\in \Delta$ such that $(x,z) \in R^\mathcal{I}$ and $(z,y) \in w_1^\mathcal{I}\}$.

\begin{definition}
An interpretation $\mathcal{I}$ \textbf{satisfies} a role inclusion
axiom $w \subseteq R$, written $\mathcal{I} \models w \subseteq R'$, if
$w^{\mathcal{I}} \subseteq R^{\mathcal{I}}$. An interpretation is a
\textbf{model} of a role hierarchy $\mathcal{R}_h$ if it satisfies all
role inclusion axioms in $\mathcal{R}_h$, written $\mathcal{I} \models
\mathcal{R}_h$.

As for role assertion axioms, for each interpretation $\mathcal{I}$ and all x, y, z $\in \Delta^\mathcal{I}$, we have:\\
\begin{tabular}{ll}
$\mathcal{I} \models Asy(S) \text{ iff } (x,y) \in S^\mathcal{I} \Rightarrow (y,x) \not\in S^\mathcal{I}$ & \mbox{(role asymmetry)}\\
$\mathcal{I} \models Ref(S) \text{ iff } Diag^\mathcal{I} \subseteq S^\mathcal{I}$ & \mbox{(role reflexivity)}\\
$\mathcal{I} \models Irr(S) \text{ iff } Diag^\mathcal{I} \cap S^\mathcal{I} = \emptyset$ & \mbox{(role irefflexivity)}\\
$\mathcal{I} \models Dis(S_1,S_2) \text{ iff } S_1^\mathcal{I} \cap S_2^\mathcal{I} = \emptyset$ & \mbox{(role disjunction)}\\
\end{tabular} 

\noindent where $Diag^\mathcal{I} = \{(x,x) | x \in \Delta^\mathcal{I}\}$.
\end{definition}
\begin{definition}
For each interpretation $\mathcal{I}$, the valuations of
$\mathcal{SROIQ}$ concepts are defined as:
\begin{itemize}
\item $\bot^\mathcal{I} = \emptyset$
\item $(\neg\;C)^\mathcal{I} = \Delta^\mathcal{I}\backslash C^\mathcal{I}$
\item $(C\; \sqcap\; D)^\mathcal{I} = C^\mathcal{I} \cap D^\mathcal{I}$
\item $(C\; \sqcup\; D)^\mathcal{I} = C^\mathcal{I} \cup D^\mathcal{I}$
\item $(\geq n\; S\; C)^\mathcal{I} = \{x\; |\; card\{y\; |\; (x,y) \in S^\mathcal{I} \wedge y \in C^\mathcal{I}\} \geq n\}$
\item $(< n\; S\; C)^\mathcal{I} = \{x\; |\; card\{y\; |\; (x,y) \in S^\mathcal{I} \wedge y \in C^\mathcal{I}\} < n\}$
\item $(\exists R. C)^\mathcal{I} = \{x\; |\; \exists y. (x,y) \in R^\mathcal{I} \wedge y \in C^\mathcal{I}\}$
\item $(\forall R. C)^\mathcal{I} = \{x\; |\; \forall y. (x,y) \in R^\mathcal{I} \Rightarrow y \in C^\mathcal{I}\}$
\item $(\exists R. Self)^\mathcal{I} = \{x\; |\; (x,x) \in R^\mathcal{I}\}$
\end{itemize}
The valuation for concept name and nominal are not repeated as they
have been previously defined in \defref{def:defintnames}.
\end{definition}

\begin{example}
The valuation of $\mathfrak{C}$ is then
$\{x\; |\; (\texttt{Alice}^\mathcal{I} = x \wedge \exists y. ((x,y) \in
\textsc{Brother}^\mathcal{I} \wedge \texttt{Bob}^\mathcal{I} = y)
\wedge \forall z. ((z,x) \in \textsc{Sister}^\mathcal{I} \Rightarrow z
\in \textbf{Male}^\mathcal{I})) \vee (card\{a\; |\; (x,a) \in
\textsc{Parent}^\mathcal{I}\} < 3 \wedge card\{b\; |\; (b,x) \in
\textsc{Owner}^\mathcal{I} \wedge b \in \textbf{Animal}^\mathcal{I}\}
\geq 1 \wedge (x,x) \in \textsc{FamilyMember}^\mathcal{I}) \}$.
\end{example}

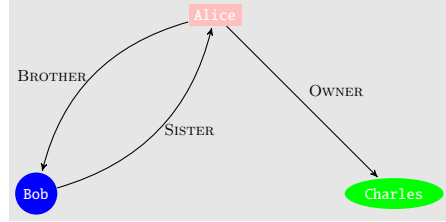
\begin{figure}
\def\smallscale{0.8}
\begin{center}
\resizebox{6cm}{!}{
\begin{tikzpicture}[->,>=stealth',shorten >=1pt,auto,node distance=5cm,
                    semithick,group/.style ={fill=gray!20, node distance=30mm},thickline/.style ={draw, thick, -latex'}][scale = 0.2]
  \tikzstyle{Male}=[circle,fill=blue,draw=none,text=white]
  \tikzstyle{Female}=[fill=pink,draw=none,text=white]
  \tikzstyle{Animal}=[ellipse,fill=green,draw=none,text=white]

  \node[Female] (A)                     {\texttt{Alice}};
  \node[Male] (B)  [below left of=A]  {\texttt{Bob}};
  \node[Animal] (C)  [below right of=A] {\texttt{Charles}};

\begin{pgfonlayer}{background}
\node [group, fit=(A) (B) (C)] (b) {};
\end{pgfonlayer}

\path (A)  edge [bend right] node[left] {\textsc{Brother}} (B)
             edge              node {\textsc{Owner}} (C)
             (B) edge [bend right] node[right] {\textsc{Sister}} (A);
\begin{scope}
\end{scope}
\end{tikzpicture}
}
\end{center}
\caption{Example of a model of $\mathfrak{C}$.
Pink rectangle nodes are \textbf{Female}, round blue nodes are
\textbf{Male} and green elliptic nodes are \textbf{Animal}.}\label{fig:exC}
\end{figure}

\begin{definition}
Let $c$ be a concept name, $i$ and $j$ be individuals, $C$ be a concept and
$r$ be a role name, the interpretations of concepts with substitutions are: 
\begin{itemize}
\item $(C \epsilon)^\mathcal{I} = C^\mathcal{I}$
\item $(C [c := c + i])^\mathcal{I} = C^\mathcal{I}$ where
$c^\mathcal{I}$ is replaced by $c^\mathcal{I} \cup i^\mathcal{I}$
\item $(C [c := c - i])^\mathcal{I} = C^\mathcal{I}$ where
$c^\mathcal{I}$ is replaced by $c^\mathcal{I} \cap \overline{i^\mathcal{I}}$
\item $(C [r := r + (i, j)])^\mathcal{I} = C^\mathcal{I}$ where
$r^\mathcal{I}$ is replaced by $r^\mathcal{I} \cup \{(i^\mathcal{I},j^\mathcal{I})\}$
\item $(C [r := r - (i, j)])^\mathcal{I} = C^\mathcal{I}$ where
$r^\mathcal{I}$ is replaced by $r^\mathcal{I} \cap \overline{\{(i^\mathcal{I}, j^\mathcal{I})\}}$
\end{itemize}
\end{definition}

\begin{example}
The valuation of ($\exists
\textsc{Sister}.\textbf{Female}$)[\textbf{Female} := \textbf{Female} +
\textsf{Alice}] is then $\{x\; |\; \exists y. ((x,y) \in
\textsc{Sister}^\mathcal{I} \wedge y \in (\textbf{Female}^\mathcal{I}
\cup \textsf{Alice}^\mathcal{I})\}$.
\end{example}

\section{Decidability of the satisfiability problem in $\mathcal{SROIQ}^\sigma$}\label{sec:decidability}

It is known that the satisfiability problem of $\mathcal{SROIQ}$
is decidable \cite{horrocks_kutz_sattler_kr_2006}. We show below that this nice
property still holds when considering explicit substitutions.

\begin{theorem}\label{theo:Decidability}
Let $\Phi_0$ be a $\mathcal{SROIQ^\sigma}$ concept, $\mathcal{R}_h$ be
a regular role hierarchy and $\mathcal{R}_a$ be
a finite set of role assertions. The satisfiability of $\Phi_0$
w.r.t. $\mathcal{R}_h$ and $\mathcal{R}_a$ is decidable.
\end{theorem}

The proof of \theoremref{theo:Decidability} is done by
translating concepts with substitutions into concepts without
substitutions.


The idea is to associate to each possible pair of a concept and a
substitution an equivalent concept not containing any substitution. In
the following, these pairs are grouped by concept constructor as two different
substitutions affecting the same concept constructor are often treated the
same. From now on, 
let $c$ and $c'$ be different concept names, $C$ and $D$ be concepts, $R$
and $R'$ be role names such that $R \neq R'$, $S$ be a simple role
name such
that $S \neq R'$,
$o$ be a nominal, $\theta$ be a substitution,
$o_i$ be a nominal associated to individual $i$ such that $o_i^\mathcal{I} = i^\mathcal{I}$ and $\bowtie$
be either $<$ or $\geq$. For ease of reading, we define $C \Rightarrow
D$ as $\neg C \sqcup D$. 

The translation of concepts with substitutions into concepts without
substitutions is defined by means of a system of 41 rules, that we
name $\mathcal{T}$, displayed below.

\begin{description}
\item[1] $\bot\; \theta \leadsto \bot$ 
\item[2] $o\; \theta \leadsto o$
\end{description}
Since substitutions do not affect the valuations of $\bot$ and the
nominals, viz. $o$, the valuations are the same with or without substitutions.
\begin{description}  
\item[3] $c [R := R \pm (i,j)] \leadsto c$
\item[4] $c [c' := c' \pm i] \leadsto c$
\item[5] $c [c := c + i] \leadsto c \sqcup o_i$
\item[6] $c [c := c - i] \leadsto c \sqcap \neg o_i$
\end{description}
The substitutions $[R := R \pm (i,j)]$ and $[c' := c' \pm i]$ do not
affect the valuation of concept name $c$ thus the valuation is the
same with or without them. 
On the other hand, if one adds element $i$ to
(resp. removes it from) the
valuation of concept name $c$, viz. rule 5 (resp rule 6), then an element will be in the valuation of $c$
iff it was before or if it is $i$ itself (resp. it was before and it is not
$i$).
\begin{description}
\item[7] $(\neg C)\; \theta \leadsto \neg(C\; \theta)$
\item[8] $(C \sqcup D)\;\theta \leadsto C\; \theta \sqcup D\; \theta$
\item[9] $(C \sqcap D)\; \theta \leadsto C\; \theta \sqcap D\; \theta$
\end{description}
Substitutions are propagated along boolean operators in an obvious
way.
\begin{description}
\item[10] $\exists R.Self [c := c \pm i] \leadsto \exists
  R.Self$
\item[11] $\exists R^-.Self [c := c \pm i] \leadsto \exists
  R^-.Self$
\item[12] $\exists R.Self [R' := R' \pm (i,j)] \leadsto \exists
  R.Self$
\item[13] $\exists R^-.Self [R' := R' \pm (i,j)] \leadsto \exists
  R^-.Self$
\end{description}
Substitutions $[c := c \pm i]$ and $[R' := R' \pm (i,j)]$ do not
affect the role name $R$, hence the rules 10 -- 14.
\begin{description}
\item[14] $\exists R.Self [R := R + (i,j)] \leadsto (o_i \sqcap
  o_j) \sqcup \exists R.Self$
\item[15] $\exists R.Self [R := R - (i,j)] \leadsto (\neg o_i \sqcup
 \neg o_j) \sqcap \exists R.Self$
\item[16] $\exists R^-.Self [R := R - (i,j)] \leadsto (\neg o_i \sqcup
 \neg o_j) \sqcap \exists R^-.Self$
\item[17] $\exists R^-.Self [R := R + (i,j)] \leadsto (o_i \sqcap
  o_j) \sqcup \exists R^-.Self$
\end{description}
$\exists R.Self$ is satisfied by an element, say $k$, when adding $(i,j)$ to (resp. removing
from) $R$ iff it was already satisfied or $k = i = j$
(resp. it was already satisfied and either $k \neq i$ or
$k \neq j$). The direction of the self-loop being irrelevant, the
translations are the same for $\exists R^-.Self$.
\begin{description}
\item[18]$(\bowtie\; n\; S\; C) [c := c \pm i] \leadsto
  (\bowtie\; n\; S\; C [c := c \pm i]) $
\item[19] $(\bowtie\; n\; S^-\; C) [c := c \pm i] \leadsto
  (\bowtie\; n\; S^-\; C [c := c \pm i]) $
\item[20] $(\bowtie\; n\; S\; C) [R' := R' \pm (i,j)] \leadsto
  (\bowtie\; n\; S\; C [R' := R' \pm (i,j)]) $
\item[21] $(\bowtie\; n\; S^-\; C) [R' := R' \pm (i,j)] \leadsto
  (\bowtie\; n\; S^-\; C [R' := R' \pm (i,j)]) $
\end{description}
Substitutions $[c := c \pm i]$ and $[R' := R' \pm (i,j)]$ do not
modify the valuation of $S$, hence the rules 18 -- 21.
\begin{description}
\item[22] 
$(\bowtie\; n\; S\; C) [S := S + (i,j)] \leadsto$
\vspace{-0.2cm}\begin{tabbing}
\hspace{0.5cm} \= \hspace{0.2cm} \= $((o_i \sqcap \exists U. (o_j\; \sqcap C [S := S + (i,j)])  \sqcap \forall S. \neg o_j)$ \hspace{0.3cm}\= $\Rightarrow$ \\
  \> \>  $(\bowtie\; (n-1)\; S\; C [S := S + (i,j)]) )$\\
  \> $\sqcap$ \> $((\neg o_i \sqcup \forall U. (\neg o_j
   \sqcup \neg C [S := S + (i,j)]) \sqcup \exists S. o_j)$ \>
   $\Rightarrow$\\
  \> \> $(\bowtie\; n\; S\; C [S := S + (i,j)]))$
\end{tabbing}
\item[23] 
$(\bowtie\; n\; S^-\; C) [S := S + (i,j)] \leadsto$
\vspace{-0.2cm}\begin{tabbing}
\hspace{0.5cm} \= \hspace{0.3cm}\= $((o_j \sqcap \exists U. (o_i\; \sqcap C [S := S + (i,j)])
 \sqcap \forall S^-. \neg o_i)$ \hspace{0.3cm}\= $\Rightarrow$ \\
 \> \>  $(\bowtie\; (n-1)\; S^-\; C [S := S + (i,j)]) )$\\
 \> $\sqcap$ \> $((\neg o_j \sqcup \forall U. (\neg o_i
  \sqcup \neg C [S := S + (i,j)]) \sqcup \exists S^-. o_i)$ \>
  $\Rightarrow$\\
 \> \> $(\bowtie\; n\; S^-\; C [S := S + (i,j)]))$
\end{tabbing}
\item[24] 
$(\bowtie\; n\; S\; C) [S := S - (i,j)] \leadsto$
\vspace{-0.2cm}\begin{tabbing}
\hspace{0.5cm} \= \hspace{0.2cm}
\= $((o_i \sqcap \exists U. (o_j\; \sqcap C [S := S - (i,j)])
 \sqcap \exists S. o_j)$ \hspace{0.8cm}\= $\Rightarrow$ \\
 \> \>  $(\bowtie\; (n+1)\; S\; C [S := S - (i,j)]) )$\\
 \> $\sqcap$ \> $((\neg o_i \sqcup \forall U. (\neg o_j
  \sqcup \neg C [S := S - (i,j)]) \sqcup \forall S. \neg o_j)$ \>
  $\Rightarrow$\\
 \> \> $(\bowtie\; n\; S\; C [S := S - (i,j)]))$
\end{tabbing}
\item[25] 
$(\bowtie\; n\; S^-\; C) [S := S - (i,j)] \leadsto$
\vspace{-0.2cm}\begin{tabbing}
\hspace{0.5cm} \= \hspace{0.2cm}
\=
  $((o_j \sqcap \exists U. (o_i\; \sqcap C [S := S - (i,j)])
 \sqcap \exists S^-. o_i)$ \hspace{0.7cm}\= $\Rightarrow$ \\
 \> \>  $(\bowtie\; (n+1)\; S^-\; C [S := S - (i,j)]) )$\\
 \> $\sqcap$ \> $((\neg o_j \sqcup \forall U. (\neg o_i
  \sqcup \neg C [S := S - (i,j)]) \sqcup \forall S^-. \neg o_i)$ \>
  $\Rightarrow$\\
 \> \> $(\bowtie\; n\; S^-\; C [S := S - (i,j)]))$
\end{tabbing}
\end{description}
The concepts in the left-hand sides of the rules above are satisfied by an element, say $k$, if three
conditions are met (see \figref{fig:ex_count}). First, $k$ must be
such that $k = i$ (resp. $k = j$ if the concept uses $S^-$). Second, the
addition (resp. removal) should change something, that is the edge $(i,j)$
should not already exist (resp. should already exist). Third, the
valuation of the concept is modified if $j$ (resp. $i$ if the concept uses $S^-$) satisfies $C$ after
the substitution. If one of these conditions is not
met, one just needs to count the neighbors satisfying $C$ after the
substitution. If all are met, there is exactly one more (resp. one
less) element after the substitution.
\begin{description}
\item[26] $(\exists R .C)[c := c \pm i] \leadsto \exists R .(C[c := c \pm
  i])$
\item[27] $(\exists R^-.C)[c := c \pm i] \leadsto \exists R .(C[c := c \pm
  i])$
\item[28] $(\exists R .C)[R' := R' \pm (i,j)] \leadsto \exists R .(C[R' :=
  R' \pm (i,j)])$
\item[29] $(\exists R^-.C)[R' := R' \pm (i,j)] \leadsto \exists R .(C[R'
  := R' \pm (i,j)])$
\end{description}
Substitutions $[c := c \pm i]$ and $[R' := R' \pm (i,j)]$ do not
modify the valuation of $R$, hence the rules 26 -- 29.
\begin{description}
\item[30] $(\exists R .C)[R := R + (i,j)] \leadsto \\
 (o_i \Rightarrow \exists U. (o_j
  \sqcap C [R := R + (i,j)]) \sqcup \exists R. C [R := R + (i,j)])
  \sqcap\\
(\neg o_i \Rightarrow \exists R .(C[R := R + (i,j)]))$
\item[31] $(\exists R^-.C)[R := R + (i,j)] \leadsto \\
 (o_j \Rightarrow \exists U. (o_i
  \sqcap C [R := R + (i,j)]) \sqcup \exists R^-. C [R := R + (i,j)])
  \sqcap\\
(\neg o_j \Rightarrow \exists R^-.(C[R := R + (i,j)]))$
\item[32] \begin{tabbing}
$(\exists R .C)[R := R - (i,j)]$ \= $\leadsto$ \=
 $(o_i \Rightarrow \exists R. (C [R := R - (i,j)] \sqcap \neg o_j))$\\
\> $\sqcap$ \>
$ (\neg o_i \Rightarrow \exists R. (C [R := R - (i,j)]))$
\end{tabbing}
\item[33] \begin{tabbing}
$(\exists R^-.C)[R := R - (i,j)]$ \= $\leadsto$ \=
 $(o_j \Rightarrow \exists R^-. (C [R := R - (i,j)] \sqcap \neg
 o_i))$\\
 \> $\sqcap$ \>
$ (\neg o_j \Rightarrow \exists R^-. (C [R := R - (i,j)]))$
\end{tabbing}
\end{description}
In these cases, if the considered element, say $k$ is not $i$
(resp. not $j$ when
$R^-$ is used), one only forwards the substitution as the new edge (
$(i,j)$ )
(resp. the suppressed edge) is
not used. If $k=i$ (resp. $k=j$), then either the new edge $(i,j)$
is used and then $j$ (resp. $i$) must satisfy $C$ after the
substitution or it is not used and then the substitution is forwarded. If
one deletes the edge $(i,j)$, then one has to check that $j$
(resp. $i$) was not
the only element satisfying $C$ after the substitution, that is to say,
there is an element different from $j$ (resp. $i$)
that satisfies $C$ after the substitution.
\begin{description}
\item[34] $(\forall R .C)[c := c \pm i] \leadsto \forall R .(C[c := c \pm
  i])$
\item[35] $(\forall R^-.C)[c := c \pm i] \leadsto \forall R^-.(C[c := c \pm
  i])$
\item[36] $(\forall R .C)[R' := R' \pm (i,j)] \leadsto \forall R .(C[R' :=
  R' \pm (i,j)])$
\item[37] $(\forall R^-.C)[R' := R' \pm (i,j)] \leadsto \forall R^-.(C[R'
  := R' \pm (i,j)])$
\end{description}
Substitutions $[c := c \pm i]$ and $[R' := R' \pm (i,j)]$ do not
modify the valuation of $R$, hence the rules 34 -- 37.
\begin{description}
\item[38] $(\forall R .C)[R := R + (i,j)] \leadsto \\
(o_i \Rightarrow \forall R. (C [R := R + (i,j)]) \sqcap \exists
U. (o_j \sqcap C [R := R + (i,j)]))
  \sqcap\\
 (\neg o_i \Rightarrow \forall R. (C [R := R + (i,j)]))$ 
\item[39] $(\forall R^-.C)[R := R + (i,j)] \leadsto \\
(o_j \Rightarrow \forall R^-. (C [R := R + (i,j)]) \sqcap \exists
U. (o_i \sqcap C [R := R + (i,j)]))
  \sqcap\\
 (\neg o_j \Rightarrow \forall R^-. (C [R := R + (i,j)]))$ 
\item[40] \begin{tabbing}
$(\forall R .C)[R := R - (i,j)]$ \= $\leadsto$ \= 
 $(o_i \Rightarrow \forall R. (C [R := R - (i,j)] \sqcup o_j))$\\
\> $\sqcap$ \> $(\neg o_i \Rightarrow \forall R. (C[R := R - (i,j)]))
$
\end{tabbing}
\item[41] \begin{tabbing}
$(\forall R^-.C)[R := R - (i,j)]$\= $\leadsto$ \=
 $(o_j \Rightarrow \forall R^-. (C [R := R - (i,j)] \sqcup o_i))$\\
 \>$\sqcap$\>
$(\neg o_j \Rightarrow \forall R^-. (C[R := R - (i,j)]))
$
\end{tabbing}
\end{description}
See \figref{fig:ex_forall} for the illustration. If an element, say $k$,
satisfies the left-hand side concepts of the rules above then either $ k \neq i$
(resp. $ k \neq j$ when
$R^-$ is used), which means that the substitution does not
affect any used edge, and thus all elements reachable from $k$
through an $R$-edge (resp. $R^-$-edge) satisfy $C$ after the
substitution. On the other hand, if $k = i$
(resp. $k = j$) then, when adding the new edge $(i,j)$, $j$ (resp. $i$) and all other elements reachable
from $k$ through an $R$-edge (resp. $R^-$-edge)
satisfy $C$ after the substitution. Otherwise, that is when removing the edge $(i,j)$,
the only element reachable from $k$
through an $R$-edge (resp. an $R^-$-edge) 
possibly not satisfying $C$ after the substitution is $j$ (resp. $i$).

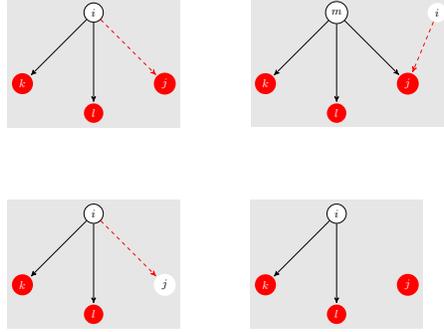
\begin{figure}
\def\smallscale{0.8}
\begin{center}
\resizebox{6cm}{!}{
\begin{tikzpicture}[->,>=stealth',shorten >=1pt,auto,node distance=2.8cm,
                    semithick,group/.style ={fill=gray!20, node distance=20mm},thickline/.style ={draw, thick, -latex'}]
  \tikzstyle{Cstate}=[circle,fill=red,draw=none,text=white]
  \tikzstyle{NCstate}=[circle,fill=white,draw=none,text=black]
  \tikzstyle{CurrState}=[circle,fill=white,draw=black,text=black]

  \node[CurrState] (I)                     {$i$};
  \node[Cstate] (K)  [below left of=I]{$k$};
  \node[Cstate] (L)  [below of=I]{$l$};
  \node[Cstate] (J)  [below right of=I] {$j$};
  \node[Cstate] (K') [right of=J]       {$k$};
  \node[CurrState] (M') [above right of=K']{$m$};
  \node[Cstate] (J') [below right of=M']{$j$};
  \node[Cstate] (L')  [below of=M']{$l$};
  \node[NCstate] (I') [right of=M']{$i$};
  \node[CurrState] (I'') [below of=L]       {$i$};
  \node[Cstate] (K'') [below left of=I'']{$k$};
  \node[Cstate] (L'')  [below of=I'']{$l$};
  \node[NCstate] (J'') [below right of=I'']{$j$};
  \node[Cstate] (K''') [right of=J'']{$k$};
  \node[CurrState] (I''') [above right of=K''']       {$i$};
  \node[Cstate] (L''')  [below of=I''']{$l$};
  \node[Cstate] (J''') [below right of=I''']{$j$};

\begin{pgfonlayer}{background}
\node [group, fit=(I) (J) (K) (L)] (b) {};
\node [group, fit=(I') (J') (K') (L') (M')] (b') {};
\node [group, fit=(I'') (J'') (K'') (L'')] (b'') {};
\node [group, fit=(I''') (J''') (K''') (L''')] (b''') {};
\end{pgfonlayer}

  \path (I)  edge node {} (L)
             edge              node {} (K)
        (M') edge node {} (K')
             edge  node {} (J')
             edge  node {} (L')
        (I'') edge node {} (K'')
             edge node {} (L'')
        (I''') edge node {} (K''')
              edge node {} (L''');
\path[color=red, dashed] (I) edge              node {} (J)
                           (I') edge              node {} (J')
                           (I'') edge node {} (J'');
\end{tikzpicture}
}
\end{center}
\caption{Example illustrating rule 24. The nodes satisfying $C[R := R
  + (i,j)]$ are drawn in red, the current node is circled in
  black. The top left graph shows the case in which the first part
  of the first implication is true, that is $(\geq 2 R.C)[R :=  R -
  (i,j)]$ is true if $(\geq 3 R. C[R:= R - (i,j)])$ is. The top right
  graph (resp. bottom left, bottom right)shows that if $o_i$ is false
  (resp. $j$ doesn't satisfy $C[R := R - (i,j)]$, $(i,j)$ is not an $R$-edge), then
  the modification does not affect the property.}\label{fig:ex_count}
\end{figure}

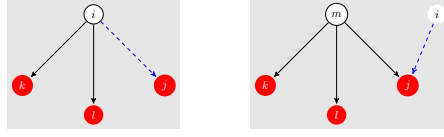
\begin{figure}
\def\smallscale{0.8}
\begin{center}
\resizebox{6cm}{!}{
\begin{tikzpicture}[->,>=stealth',shorten >=1pt,auto,node distance=2.8cm,
                    semithick,group/.style ={fill=gray!20, node distance=20mm},thickline/.style ={draw, thick, -latex'}]
  \tikzstyle{Cstate}=[circle,fill=red,draw=none,text=white]
  \tikzstyle{NCstate}=[circle,fill=white,draw=none,text=black]
  \tikzstyle{CurrState}=[circle,fill=white,draw=black,text=black]

  \node[CurrState] (I)                     {$i$};
  \node[Cstate] (K)  [below left of=I]{$k$};
  \node[Cstate] (L)  [below of=I]{$l$};
  \node[Cstate] (J)  [below right of=I] {$j$};
  \node[Cstate] (K') [right of=J]       {$k$};
  \node[CurrState] (M') [above right of=K']{$m$};
  \node[Cstate] (J') [below right of=M']{$j$};
  \node[Cstate] (L')  [below of=M']{$l$};
  \node[NCstate] (I') [right of=M']{$i$};

\begin{pgfonlayer}{background}
\node [group, fit=(I) (J) (K) (L)] (b) {};
\node [group, fit=(I') (J') (K') (L') (M')] (b') {};
\end{pgfonlayer}

  \path (I)  edge node {} (L)
             edge              node {} (K)
        (M') edge node {} (K')
             edge  node {} (J')
             edge  node {} (L');
\path[color=blue, dashed] (I) edge              node {} (J)
                           (I') edge              node {} (J');

\end{tikzpicture}
}
\end{center}
\caption{Example illustrating rule 38. The nodes satisfying $C[R := R
  + (i,j)]$ are drawn in red, the current node is circled in
  black. The leftmost graph shows the case in which $o_i$ is true, and
then $(\forall R.C)[R :=  R + (i,j)]$ will be true if $j$ satisfies $C[R := R
  + (i,j)]$. The rightmost graph shows that if $o_i$ is false, then
  the modification does not affect the property.}\label{fig:ex_forall}
\end{figure}

Now we show that, for each translation rule $L \leadsto R$, the
valuations of $L$ and $R$ are the same under a given interpretation.

\begin{lemma}\label{lemma1}
Let $\Sigma$ be a signature, $\mathcal{I}$ be an interpretation over $\Sigma$, $L \leadsto R$ be one of the
above translation rules (1 -- 41), then $L^\mathcal{I} = R^\mathcal{I}$.
\end{lemma}

\begin{proof}The complete proof of the lemma is given in the appendix. \end{proof}

Now we show that, given a $\mathcal{SROIQ}^\sigma$ concept, applying
the translation rules yields a $\mathcal{SROIQ}$ concept.

\begin{lemma}\label{lemma2}
Let $\Sigma$ be a signature, $\mathcal{T}$ is terminating.
\end{lemma}
\begin{proof}
The proof is in the appendix.
\end{proof}

\begin{lemma}
Let $\Sigma$ be a signature, $\phi$ be a $\mathcal{SROIQ}^\sigma$
concept, $\psi$ be a normal form obtained from $\phi$ by rewriting the
concept $\phi$ using the rewrite rules of $\mathcal{T}$, then $\psi$
is a $\mathcal{SROIQ}$ concept. 
\end{lemma}
\begin{proof}
The final result is substitution-free as, if there was a substitution
remaining it would have the form of the left-hand side of one of the
rules hence it wouldn't be a normal form.
\end{proof}

With these three lemmata, it is easy to prove the theorem. \\
\begin{proof}
The translations allow us to obtain
$\mathcal{SROIQ}$-concepts from $\mathcal{SROIQ}^\sigma$-concepts. As
the satisfiability of an
$\mathcal{SROIQ}$-concept is known to be decidable, so is the
satisfiability of an $\mathcal{SROIQ}^\sigma$-concept.
\end{proof}

\section{Example}\label{sec:example}

In this section, we illustrate with a simple example the translation
rules given in the previous section. Consider the graph on the left of
\figref{fig:ex_ex}, one can see that element $i$ satisfies the concept $(<\; 3\; S\; (<\; 3\; S\; \top))$.
We show hereafter that $(<\; 3\; S\; (<\; 3\; S\; \top))$
still holds at element $i$ after adding the $S$-edge $(i,j)$, that is
we want to see that:\\
\begin{tabular}[h]{l|c}
$o_i$ & \mbox{At element $i$,}\\
$\sqcap\; (<\; 3\; S\; (<\; 3\; S\; \top))[S := S + (i,j)]$ &
\mbox{$(<\; 3\; S\; (<\; 3\; S\; \top))$ is satisfied}\\
 & \mbox{after substitution
  $[S := S + (i,j)]$}\\
$\sqcap\; (<\; 3\; (<\; 3\; S\; \top))$ & \mbox{$(<\; 3\; S\; (<\; 3\; S\;
  \top)))$ is satisfied}\\
$\sqcap\; \exists S. o_i$ & \mbox{there is an edge $(i,i)$}\\
$\sqcap\; \exists S. o_k$ & \mbox{there is an edge $(i,k)$}\\
$\sqcap\; \forall S. \neg o_j$ & \mbox{there is \emph{not} an edge
  $(i,j)$}\\
$\sqcap\; \exists U.(o_j \sqcap (<\;
1\; S\; \top) \sqcap \neg o_i)$ & \mbox{$j$ has no $S$-outgoing edge
  and $j \neq i$}\\
$\sqcap\; \exists U. (o_k \sqcap (<\; 1\;
S \; \top) \sqcap \neg o_i)$ & \mbox{$k$ has no $S$-outgoing edge
  and $k \neq i$}
\end{tabular}\\
 is satisfiable.

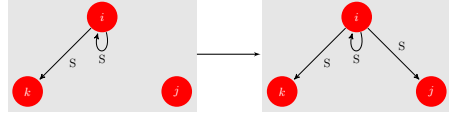
\begin{figure}
\def\smallscale{0.8}
\begin{center}
\resizebox{6cm}{!}{
\begin{tikzpicture}[->,>=stealth',shorten >=1pt,auto,node distance=2.8cm,
                    semithick,group/.style ={fill=gray!20, node distance=20mm},thickline/.style ={draw, thick, -latex'}]
  \tikzstyle{every state}=[fill=red,draw=none,text=white]

  \node[state] (I)                     {$i$};
  \node[state] (K)  [below left of=I]  {$k$};
  \node[state] (J)  [below right of=I] {$j$};
  \node[state] (K') [right of=J]       {$k$};
  \node[state] (I') [above right of=K']{$i$};
  \node[state] (J') [below right of=I']{$j$};

\begin{pgfonlayer}{background}
\node [group, fit=(I) (J) (K)] (b) {};
\node [group, fit=(I') (J') (K')] (b') {};
\end{pgfonlayer}

  \path (I)  edge [loop below] node {S} (I)
             edge              node {S} (K)
        (I') edge [loop below] node {S} (I')
             edge              node {S} (J')
             edge              node {S} (K');
\begin{scope}
\draw [thickline] (b) -- (b');
\end{scope}
\end{tikzpicture}
}
\end{center}
\caption{Example of graph transformation}\label{fig:ex_ex}
\end{figure}

By applying rule 22 to $ (<\; 3\; S\; (<\; 3\; S\; \top))[S := S +
(i,j)]$, one obtains $((o_i \sqcap \exists
 U. (o_j\; \sqcap (<\; 3\; S\; \top)[S := S + (i,j)])
   \sqcap \forall S. \neg o_j) \Rightarrow\\
  (<\; 2\; S\; (<\; 3\; S\; \top)[S := S + (i,j)]))\\
  \sqcap ((\neg o_i \sqcup \forall U. (\neg o_j
   \sqcup \neg (<\; 3\; S\; \top)[S := S + (i,j)]) \sqcup \exists S. o_j)
   \Rightarrow\\
  (<\; 3\; S\; (<\; 3\; S\; \top)[S := S + (i,j)]))$.\\
Let's deal with $ (<\; 3\; S\; \top)[S := S + (i,j)]$. By applying
rule 22 again, one gets $(o_i \sqcap \exists U. (o_j\; \sqcap \top[S := S + (i,j)])
   \sqcap \forall S. \neg o_j) \Rightarrow\\
  (<\; 2\; S\; \top [S := S + (i,j)]) )
  \sqcap ((\neg o_i \sqcup \forall U. (\neg o_j
   \sqcup \neg \top [S := S + (i,j)]) \sqcup \exists S. o_j)
   \Rightarrow\\
  (<\; 3\; S\; \top [S := S + (i,j)])$. When applying rules 1 and 7,
  it is obvious that $\top [S := S + (i,j)])$ is $\top$. \\
Then, in $\exists
 U. (o_j\; \sqcap (<\; 3\; S\; \top)[S := S + (i,j)])$, knowing that $
 \exists U.(o_j \sqcap \neg o_i)$, $\neg o_i \sqcup \forall U. (\neg o_j
   \sqcup \neg \top [S := S + (i,j)]) \sqcup \exists S. o_j$ is
   satisfied and thus $\exists
 U. (o_j\; \sqcap (<\; 3\; S\; \top)[S := S + (i,j)])$ becomes $\exists
 U. (o_j\; \sqcap (<\; 3\; S\; \top))$ which, as $\exists U.(o_j \sqcap (<\;
1\; S\; \top))$ is satisfied, is true. As $\forall U. (\neg o_j
   \sqcup \neg (<\; 3\; S\; \top)[S := S + (i,j)])$ is exactly $\neg (\exists
 U. (o_j\; \sqcap (<\; 3\; S\; \top)[S := S + (i,j)]))$, $\forall U. (\neg o_j
   \sqcup \neg (<\; 3\; S\; \top)[S := S + (i,j)])$ is false.\\
Thus, going back to $((o_i \sqcap \exists
 U. (o_j\; \sqcap (<\; 3\; S\; \top)[S := S + (i,j)])
   \sqcap \forall S. \neg o_j) \Rightarrow\\
  (<\; 2\; S\; (<\; 3\; S\; \top)[S := S + (i,j)]))\\
  \sqcap ((\neg o_i \sqcup \forall U. (\neg o_j
   \sqcup \neg (<\; 3\; S\; \top)[S := S + (i,j)]) \sqcup \exists S. o_j)
   \Rightarrow\\
  (<\; 3\; S\; (<\; 3\; S\; \top)[S := S + (i,j)]))$, what one gets is $((o_i \sqcap \top
   \sqcap \forall S. \neg o_j) \Rightarrow\\
  (<\; 2\; S\; (<\; 3\; S\; \top)[S := S + (i,j)]))
  \sqcap ((\neg o_i \sqcup \bot \sqcup \exists S. o_j)
   \Rightarrow\\
  (<\; 3\; S\; (<\; 3\; S\; \top)[S := S + (i,j)]))$.\\
As we are at $i$ and, from $\forall S. \neg o_j$, there is no $S$-edge
$(i,j)$, what we have to prove is $(<\; 2\; S\; (<\; 3\; S\; \top)[S
:= S + (i,j)]$, that is that $i$ has less than 2 $S$-neighbors such
that they have less than 3 $S$-neighbors after the substitution. We
will test all of these neighbors. From $(<\; 3\; S\; \top)$, we know
that there are at most 2 neighbors and, from $\exists S. o_i$ and
$\exists S. o_k$, we know that these neighbors are $i$ and $k$. We
have to prove that one of them doesn't satisfy $(<\; 3\; S\; \top)[S := S + (i,j)]$.\\
Let's prove that $i$ doesn't satisfy $(<\; 3\; S\; \top)[S := S + (i,j)]$. As $o_i \sqcap \exists U. (o_j\; \sqcap \top[S := S + (i,j)]) \sqcap \forall S. \neg o_j$ is satisfied at $i$ (we know that $i$
   satisfies $o_i$, that there is an element $j$ (thus $\exists
   U. (o_j\; \sqcap \top)$ is true) and that there is no $S$-edge
   $(i,j)$),  $(<\; 3\; S\; \top)[S := S + (i,j)]$ is satisfied at $i$
   iff $ (<\; 2\; S\; \top [S := S + (i,j)])$ is
   satisfied at $i$. But, as $\top [S := S + (i,j)] \leadsto \top$, $
   (<\; 2\; S\; \top [S := S + (i,j)])$ is not satisfied at $i$ (as
   $i$ has 2 different neighbors). Thus, $(<\; 2\; S\; (<\; 3\; S\; \top)[S
:= S + (i,j)]$ is satisfied at $i$ which means that the initial
concept is satisfiable. 

\section{Conclusion}\label{sec:conclusion}

We have introduced a new Description Logic named
$\mathcal{SROIQ^\sigma}$ which is an extension of
$\mathcal{SROIQ}$ with substitutions. We have proven that this logic is decidable by translating it to $\mathcal{SROIQ}$.

The same method can likely be used for several other DLs. For
instance, $\mathcal{SHOIQ}$ \cite{HorrSattSHOIQ2005} is a restriction of $\mathcal{SROIQ}$
without role assertions and $\exists R.Self$. As none of the
translation rules outside these constructs uses these constructs, they
can be removed and $\mathcal{SHOIQ}^\sigma$ can be defined in the
exact same way as $\mathcal{SROIQ}^\sigma$.

On the other hand, Description Logics like $\mathcal{ALC}$ \cite{Schmidt-SchaubB:1991:ACD:114341.114342} that is
restricted to concepts without counting quantifiers are not suited for
such a translation. Existential and universal quantification would
create couting quantifiers during translation that are outside the
scope of the logic. Even extending $\mathcal{ALC}$ with counting
quantifiers, which produces $\mathcal{ALCQ}$, is not enough as
nominals are used during translation.


\newpage
\section*{Appendix}

First we prove \lemmaref{lemma1}.

\begin{framed}
Let $\Sigma$ be a signature, $\mathcal{I}$ be an interpretation over $\Sigma$, $L \leadsto R$ be one of the
above translation rules (1 -- 41), then $L^\mathcal{I} =
R^\mathcal{I}$.
\end{framed}

\begin{proof}
\begin{enumerate}
\item $(\bot\; \theta)^\mathcal{I} = \emptyset = \bot^\mathcal{I}$ as
  it does not depend on the valuation of any concept or role.
\item As the valuation of a nominal is independent of the valuations
  of all roles and concepts, $(o\; \theta)^\mathcal{I} =
  o^\mathcal{I}$.
\item $(c [R := R \pm (i,j)])^\mathcal{I} =
  c^\mathcal{I}$ as $c
  ^\mathcal{I}$ does not depend on the valuation of $R$.
\item $(c [c' := c' \pm i])^\mathcal{I} =
  c^\mathcal{I}$ as $c^\mathcal{I}$ does not depend on the valuation of $c'$.
\item $(c [c := c + i])^\mathcal{I} = (c \sqcup
  o_i)^\mathcal{I}$. By definition of the valuation of a substitution,
  $(c [c := c + i])^\mathcal{I} = c^\mathcal{I} \cup i = (c \sqcup
  o_i)^\mathcal{I}$.
\item $(c [c := c - i])^\mathcal{I} = (c \sqcap \neg
  o_i)^\mathcal{I}$. By definition of the valuation of a substitution,
  $(c [c := c - i])^\mathcal{I} = c^\mathcal{I} \cap \overline{i} = (c \sqcap
 \neg o_i)^\mathcal{I}$.
\item By definition, $((\neg C) [c := c + i])^\mathcal{I} = (\neg
  C)^\mathcal{I}$ where $c^\mathcal{I}$ is replaced by $c^\mathcal{I}
  \cup i$. As $(\neg C)^\mathcal{I} = \Delta \cap
  \overline{C^\mathcal{I}}$, $((\neg C) [c := c + i])^\mathcal{I} =
  \Delta \cap \overline{C'^{\mathcal{I}}}$ with $C'^{\mathcal{I}} =
  C^\mathcal{I}$ where $c^\mathcal{I}$ is replaced by $c^\mathcal{I}
  \cup i$. This is exactly the definition of $(\neg (C [c := c +
  i]))^\mathcal{I}$. The same can be done with the other substitutions.
\item By definition, $((C \sqcup D) [c := c + i])^\mathcal{I} = (C \sqcup
  D)^\mathcal{I}$ where $c^\mathcal{I}$ is replaced by $c^\mathcal{I}
  \cup i$. That is $(C' \sqcup D')^\mathcal{I}$ with $C'^{\mathcal{I}} =
  C^\mathcal{I}$ where $c^\mathcal{I}$ is replaced by $c^\mathcal{I}
  \cup i$ and $D'^{\mathcal{I}} =
  D^\mathcal{I}$ where $c^\mathcal{I}$ is replaced by $c^\mathcal{I}
  \cup i$. This is exactly the definition of $(C [c := c +
  i] \sqcup D [c := c + i])^\mathcal{I}$. The same can be done with the other substitutions.
\item By definition, $((C \sqcap D) [c := c + i])^\mathcal{I} = (C \sqcap
  D)^\mathcal{I}$ where $c^\mathcal{I}$ is replaced by $c^\mathcal{I}
  \cup i$. That is $(C' \sqcap D')^\mathcal{I}$ with $C'^{\mathcal{I}} =
  C^\mathcal{I}$ where $c^\mathcal{I}$ is replaced by $c^\mathcal{I}
  \cup i$ and $D'^{\mathcal{I}} =
  D^\mathcal{I}$ where $c^\mathcal{I}$ is replaced by $c^\mathcal{I}
  \cup i$. This is exactly the definition of $(C [c := c +
  i] \sqcap D [c := c + i])^\mathcal{I}$. The same can be done with the other substitutions.
\item $(\exists R.Self [c := c \pm i])^\mathcal{I} = (\exists
  R.Self)^\mathcal{I}$ as $(\exists
  R.Self)^\mathcal{I}$ is independent of the valuation of $c$.
\item $(\exists R^-.Self [c := c \pm i])^\mathcal{I} = (\exists
  R^-.Self)^\mathcal{I}$ as $(\exists
  R^-.Self)^\mathcal{I}$ is independent of the valuation of $c$.
\item $(\exists R.Self [R' := R' \pm (i,j)])^\mathcal{I}= (\exists
  R.Self)^\mathcal{I}$ as $(\exists
  R.Self)^\mathcal{I}$ is independent of the valuation of $R'$.
\item $(\exists R^-.Self [R' := R' \pm (i,j)])^\mathcal{I}= (\exists
  R^-.Self)^\mathcal{I}$ as $(\exists
  R^-.Self)^\mathcal{I}$ is independent of the valuation of $R'$.
\item As {$(\exists R.Self)^\mathcal{I} = \{x | (x,x) \in
  R^\mathcal{I}\}$}, ${(\exists R.Self [R := R + (i,j)])^\mathcal{I}} =$\\ ${\{x | (x,x) \in
  R^\mathcal{I} \cup \{(i,j)\}\}} = (\exists R.Self)^\mathcal{I} \cup
  (\{i\} \cap \{j\})$ that is \\$(\exists R.Self [R := R + (i,j)])^\mathcal{I} = ((o_i \sqcap
  o_j) \sqcup \exists R.Self)^\mathcal{I}$.
\item As {$(\exists R^-.Self)^\mathcal{I} = \{x | (x,x) \in
  R^\mathcal{I}\}$}, ${(\exists R^-.Self [R := R + (i,j)])^\mathcal{I}} =$\\ ${\{x | (x,x) \in
  R^\mathcal{I} \cup \{(i,j)\}\}} = (\exists R^-.Self)^\mathcal{I} \cup
  (\{i\} \cap \{j\})$ that is \\$(\exists R^-.Self [R := R + (i,j)])^\mathcal{I} = ((o_i \sqcap
  o_j) \sqcup \exists R^-.Self)^\mathcal{I}$.
\item As $(\exists R.Self)^\mathcal{I} = \{x | (x,x) \in
  R^\mathcal{I}\}$, $(\exists R.Self [R := R - (i,j)])^\mathcal{I} =$\\ $\{x | (x,x) \in
  R^\mathcal{I} \cap \overline{\{(i,j)\}}\} = (\exists R.Self)^\mathcal{I} \cap
  (\overline{\{i\}} \cup \overline{\{j\}})$ that is \\$(\exists R.Self [R := R - (i,j)])^\mathcal{I} = ((\neg o_i \sqcup
 \neg o_j) \sqcap \exists R.Self)^\mathcal{I}$.
\item As $(\exists R^-.Self)^\mathcal{I} = \{x | (x,x) \in
  R^\mathcal{I}\}$, $(\exists R^-.Self [R := R - (i,j)])^\mathcal{I} =$\\ $\{x | (x,x) \in
  R^\mathcal{I} \cap \overline{\{(i,j)\}}\} = (\exists R^-.Self)^\mathcal{I} \cap
  (\overline{\{i\}} \cup \overline{\{j\}})$ that is \\$(\exists R^-.Self [R := R - (i,j)])^\mathcal{I} = ((\neg o_i \sqcup
 \neg o_j) \sqcap \exists R^-.Self)^\mathcal{I}$.
\item As $(\bowtie\; n\; S\; C)^\mathcal{I} = \{x\; |\; card\{y\; |\;
  (x,y) \in S^\mathcal{I} \wedge y \in C^\mathcal{I}\} \bowtie n\}$,\\
  $((\bowtie\; n\; S\; C) [c := c \pm i])^\mathcal{I}$ =\\$\{x\; |\; card\{y\; |\;
  (x,y) \in S^\mathcal{I} \wedge y \in C[c := c \pm i]^\mathcal{I}\}$
  $\bowtie n\}$ as $S^\mathcal{I}$ is independent of the valuation of
  $c$. Thus $((\bowtie\; n\; S\; C) [c := c \pm i])^\mathcal{I} =
  (\bowtie\; n\; S\; C [c := c \pm i])^\mathcal{I}$.
\item As $(\bowtie\; n\; S^-\; C)^\mathcal{I} = \{x\; |\; card\{y\; |\;
  (y,x) \in S^\mathcal{I} \wedge y \in C^\mathcal{I}\} \bowtie n\}$,\\
  $((\bowtie\; n\; S^-\; C) [c := c \pm i])^\mathcal{I} =$\\ $
\{x\; |\; card\{y\; |\;
  (y,x) \in S^\mathcal{I} \wedge y \in C[c := c \pm i]^\mathcal{I}\}
  \bowtie n\}$ as $S^\mathcal{I}$ is independent of the valuation of
  $c$. Thus $((\bowtie\; n\; S^-\; C) [c := c \pm i])^\mathcal{I} =
  (\bowtie\; n\; S^-\; C [c := c \pm i])^\mathcal{I}$.
\item As $(\bowtie\; n\; S\; C)^\mathcal{I} = \{x\; |\; card\{y\; |\;
  (x,y) \in S^\mathcal{I} \wedge y \in C^\mathcal{I}\} \bowtie n\}$,\\
  $((\bowtie\; n\; S\; C) [R' := R' \pm (i,j)])^\mathcal{I} =$\\ $\{x\; |\; card\{y\; |\;
  (x,y) \in S^\mathcal{I} \wedge y \in C[R' := R' \pm (i,j)]^\mathcal{I}\}
  \bowtie n\}$ as $S^\mathcal{I}$ is independent of the valuation of
  $R'$. Thus\\ $((\bowtie\; n\; S\; C) [R' := R' \pm (i,j)])^\mathcal{I} =
 (\bowtie\; n\; S\; C [R' := R' \pm (i,j)])^\mathcal{I}$.
\item As $(\bowtie\; n\; S^-\; C)^\mathcal{I} = \{x\; |\; card\{y\; |\;
  (y,x) \in S^\mathcal{I} \wedge y \in C^\mathcal{I}\} \bowtie n\}$,\\
  $((\bowtie\; n\; S^-\; C) [R' := R' \pm (i,j)])^\mathcal{I} =$\\ $\{x\; |\; card\{y\; |\;
  (y,x) \in S^\mathcal{I} \wedge y \in C[R' := R' \pm (i,j)]^\mathcal{I}\}
  \bowtie n\}$ as $S^\mathcal{I}$ is independent of the valuation of
  $R'$. Thus\\ $((\bowtie\; n\; S^-\; C) [R' := R' \pm (i,j)])^\mathcal{I} =
 (\bowtie\; n\; S^-\; C [R' := R' \pm (i,j)])^\mathcal{I}$.
\item As $(\bowtie\; n\; S\; C)^\mathcal{I} = \{x\; |\; card\{y\; |\;
  (x,y) \in S^\mathcal{I} \wedge y \in C^\mathcal{I}\} \bowtie n\}$,\\
  $((\bowtie\; n\; S\; C)[S := S + (i,j)])^\mathcal{I} =$\\ $\{x\; |\; card\{y\; |\;
  (x,y) \in S [S := S + (i,j)]^\mathcal{I} \wedge y \in C [S := S +
  (i,j)]^\mathcal{I}\} \bowtie n\} =$\\ $\{x\; |\; card\{y\; |\;
  (x,y) \in S^\mathcal{I} \cup \{(i,j)\} \wedge y \in C [S := S +
  (i,j)]^\mathcal{I}\} \bowtie n\}$. We consider $\{y\; |\;
  (x,y) \in S^\mathcal{I} \cup \{(i,j)\} \wedge y \in C [S := S +
  (i,j)]^\mathcal{I}\}$ and try the possible sets:
\begin{itemize}
\item If $x \neq i$ or $j \not\in C [S := S +
  (i,j)]^\mathcal{I}$ or $(i,j) \in S^\mathcal{I}$, then \\$\{y\; |\;
  (x,y) \in S^\mathcal{I} \cup \{(i,j)\} \wedge y \in C [S := S +
  (i,j)]^\mathcal{I}\} =$\\ $\{y\; |\;
  (x,y) \in S^\mathcal{I} \wedge y \in C [S := S +
  (i,j)]^\mathcal{I}\}$,
\item else, $x = i$ and $j \in C [S := S +
  (i,j)]^\mathcal{I}$ and $(i,j) \not\in S^\mathcal{I}$, and thus $\{y\; |\;
  (x,y) \in S^\mathcal{I} \cup \{(i,j)\} \wedge y \in C [S := S +
  (i,j)]^\mathcal{I}\} =$\\ $\{y\; |\;
  (x,y) \in S^\mathcal{I} \wedge y \in C [S := S +
  (i,j)]^\mathcal{I}\} \cup \{(i,j)\}$. As $\{(i,j)\}$ is disjoint
  from $ \{y\; |\;
  (x,y) \in S^\mathcal{I} \wedge y \in C [S := S +
  (i,j)]^\mathcal{I}\}$, the cardinality of \\$\{y\; |\;
  (x,y) \in S^\mathcal{I} \cup \{(i,j)\} \wedge y \in C [S := S +
  (i,j)]^\mathcal{I}\}$ is exactly the cardinality of $\{y\; |\;
  (x,y) \in S^\mathcal{I} \wedge y \in C [S := S +
  (i,j)]^\mathcal{I}\} + 1$.
\item Thus, \\${\{x\; |\; card\{y\; |\;
  (x,y) \in S [S := S + (i,j)]^\mathcal{I} \wedge y \in C [S := S +
  (i,j)]^\mathcal{I}\} \bowtie n\}}$\\ $= \{x\; | $\\$(x \neq i \vee j \not\in C [S := S +
  (i,j)]^\mathcal{I} \vee (i,j) \in S^\mathcal{I} \Rightarrow$\\ $\; card\{y\; |\;
  (x,y) \in S^\mathcal{I} \wedge y \in C [S := S +
  (i,j)]^\mathcal{I}\} \bowtie n)$\\ $\wedge (x = i \wedge j \in C [S := S +
  (i,j)]^\mathcal{I} \wedge (i,j) \not\in S^\mathcal{I} \Rightarrow$\\ $card\{y\; |\;
  (x,y) \in S^\mathcal{I} \wedge y \in C [S := S +
  (i,j)]^\mathcal{I}\} \bowtie (n - 1))\}$
\end{itemize}
Thus,
$((\bowtie\; n\; S\; C) [S := S + (i,j)])^\mathcal{I} =\\
  (((o_i \sqcap \exists U. (o_j\; \sqcap C [S := S + (i,j)])
  \sqcap \forall S. \neg o_j) \Rightarrow \\
 (\bowtie\; (n-1)\; S\; C [S := S + (i,j)]) )\\
 \sqcap ((\neg o_i \sqcup \forall U. (\neg o_j
  \sqcup \neg C [S := S + (i,j)]) \sqcup \exists S. o_j)
  \Rightarrow\\
 (\bowtie\; n\; S\; C [S := S + (i,j)])))^\mathcal{I}$
\item One can see that $(\bowtie\; n\; S^-\; C)[S := S + (i,j)]$ is
  similar to\\ $(\bowtie\; n\; S^-\; C)[S^- := S^- + (j,i)]$, that is
  replacing $S$ by $S^-$ and swapping $i$ and $j$ in the previous case.
\item As $(\bowtie\; n\; S\; C)^\mathcal{I} = \{x\; |\; card\{y\; |\;
  (x,y) \in S^\mathcal{I} \wedge y \in C^\mathcal{I}\} \bowtie n\}$,\\
  $((\bowtie\; n\; S\; C)[S := S - (i,j)])^\mathcal{I} =\\ \{x\; |\; card\{y\; |\;
  (x,y) \in S [S := S - (i,j)]^\mathcal{I} \wedge y \in C [S := S -
  (i,j)]^\mathcal{I}\} \bowtie n\} =\\ \{x\; |\; card\{y\; |\;
  (x,y) \in S^\mathcal{I} \cap \overline{\{(i,j)\}} \wedge y \in C [S := S -
  (i,j)]^\mathcal{I}\} \bowtie n\}$. We consider $\{y\; |\;
  (x,y) \in S^\mathcal{I} \cap \overline{\{(i,j)}\} \wedge y \in C [S := S +
  (i,j)]^\mathcal{I}\}$ and try the possible sets:
\begin{itemize}
\item If $x \neq i$ or $j \not\in C [S := S -
  (i,j)]^\mathcal{I}$ or $(i,j) \not\in S^\mathcal{I}$, then\\ $\{y\; |\;
  (x,y) \in S^\mathcal{I} \cap \overline{\{(i,j)\}} \wedge y \in C [S := S -
  (i,j)]^\mathcal{I}\} =\\ \{y\; |\;
  (x,y) \in S^\mathcal{I} \wedge y \in C [S := S -
  (i,j)]^\mathcal{I}\}$,
\item else, $x = i$ and $j \in C [S := S -
  (i,j)]^\mathcal{I}$ and $(i,j) \in S^\mathcal{I}$, and thus $\{y\; |\;
  (x,y) \in S^\mathcal{I} \cap \overline{\{(i,j)\}} \wedge y \in C [S := S -
  (i,j)]^\mathcal{I}\} =\\ \{y\; |\;
  (x,y) \in S^\mathcal{I} \wedge y \in C [S := S -
  (i,j)]^\mathcal{I}\} \cap \overline{\{i,j\}}$. As $\{(i,j)\}
  \subseteq \\\{y\; |\;
  (x,y) \in S^\mathcal{I} \wedge y \in C [S := S -
  (i,j)]^\mathcal{I}\}$, the cardinality of \\$\{y\; |\;
  (x,y) \in S^\mathcal{I} \cap \overline{\{(i,j)}\} \wedge y \in C [S := S +
  (i,j)]^\mathcal{I}\}$ is exactly the cardinality of $\{y\; |\;
  (x,y) \in S^\mathcal{I} \wedge y \in C [S := S -
  (i,j)]^\mathcal{I}\} - 1$.
\item Thus, \\${\{x\; |\; card\{y\; |\;
  (x,y) \in S [S := S - (i,j)]^\mathcal{I} \wedge y \in C [S := S -
  (i,j)]^\mathcal{I}\} \bowtie n\}}\\ = \{x\; |\\ (x \neq i \vee j \not\in C [S := S -
  (i,j)]^\mathcal{I} \vee (i,j) \not\in S^\mathcal{I} \Rightarrow\\  card\{y\; |\;
  (x,y) \in S^\mathcal{I} \wedge y \in C [S := S -
  (i,j)]^\mathcal{I}\} \bowtie n)\\ \wedge (x = i \wedge j \in C [S := S -
  (i,j)]^\mathcal{I} \wedge (i,j) \in S^\mathcal{I} \Rightarrow\\  card\{y\; |\;
  (x,y) \in S^\mathcal{I} \wedge y \in C [S := S -
  (i,j)]^\mathcal{I}\} \bowtie (n + 1))\}$
\end{itemize}
Thus, $((\bowtie\; n\; S\; C) [S := S - (i,j)])^\mathcal{I} =\\
  (((o_i \sqcap \exists U. (o_j\; \sqcap C [S := S - (i,j)])
  \sqcap \exists S. o_j) \Rightarrow \\
 (\bowtie\; (n+1)\; S\; C [S := S - (i,j)]))\\
 \sqcap ((\neg o_i \sqcup \forall U. (\neg o_j
  \sqcup \neg C [S := S - (i,j)]) \sqcup \forall S. \neg o_j)
  \Rightarrow\\
 (\bowtie\; n\; S\; C [S := S - (i,j)])))^\mathcal{I}$
\item One can see that $(\bowtie\; n\; S^-\; C)[S := S - (i,j)]$ is
  similar to\\ $(\bowtie\; n\; S^-\; C)[S^- := S^- - (j,i)]$, that is
  replacing $S$ by $S^-$ and swapping $i$ and $j$ in the previous case.
 \item As $(\exists R. C)^\mathcal{I} = \{x\; |\; \exists y. (x,y) \in
   R^\mathcal{I} \wedge y \in C^\mathcal{I}\}$,\\
   $((\exists R. C) [c := c \pm i])^\mathcal{I} = \{x\; |\; \exists y. (x,y) \in
   R^\mathcal{I} \wedge y \in C[c := c \pm i]^\mathcal{I}\}$ as $R^\mathcal{I}$ is independent of the valuation of
   $c$. Thus \\$((\exists R. C) [c := c \pm i])^\mathcal{I} =
   (\exists R. (C [c := c \pm i]))^\mathcal{I}$.
 \item As $(\exists R^-. C)^\mathcal{I} = \{x\; |\; \exists y. (y,x) \in
   R^\mathcal{I} \wedge y \in C^\mathcal{I}\}$,\\
   $((\exists R^-. C) [c := c \pm i])^\mathcal{I} = \{x\; |\; \exists y. (y,x) \in
   R^\mathcal{I} \wedge y \in C[c := c \pm i]^\mathcal{I}\}$ as $R^\mathcal{I}$ is independent of the valuation of
   $c$. Thus \\$((\exists R^-. C) [c := c \pm i])^\mathcal{I} =
   (\exists R^-. (C [c := c \pm i]))^\mathcal{I}$.
 \item As $(\exists R. C)^\mathcal{I} = \{x\; |\; \exists y. (x,y) \in
   R^\mathcal{I} \wedge y \in C^\mathcal{I}\}$,\\
   $((\exists R. C) [R' := R' \pm (i,j)])^\mathcal{I} = \{x\; |\; \exists y. (x,y) \in
   R^\mathcal{I} \wedge\\ y \in C[R' := R' \pm (i,j)]^\mathcal{I}\}$ as $R^\mathcal{I}$ is independent of the valuation of
   $R'$. Thus $((\exists R .C)[R' := R' \pm (i,j)])^\mathcal{I}=
   (\exists R .(C[R' := R' \pm (i,j)]))^\mathcal{I}$
 \item As $(\exists R^-. C)^\mathcal{I} = \{x\; |\; \exists y. (y,x) \in
   R^\mathcal{I} \wedge y \in C^\mathcal{I}\}$,\\
   $((\exists R^-. C) [R' := R' \pm (i,j)])^\mathcal{I} = \{x\; |\; \exists y. (y,x) \in
   R^\mathcal{I} \wedge\\ y \in C[R' := R' \pm (i,j)]^\mathcal{I}\}$ as $R^\mathcal{I}$ is independent of the valuation of
   $R'$. Thus $((\exists R^-.C)[R' := R' \pm (i,j)])^\mathcal{I}= (\exists R^-.(C[R' := R' \pm (i,j)]))^\mathcal{I}$
 \item As $(\exists R. C)^\mathcal{I} = \{x\; |\; \exists y. (x,y) \in
   R^\mathcal{I} \wedge y \in C^\mathcal{I}\}$,\\
 $((\exists R. C)[R := R + (i,j)])^\mathcal{I} = \{x\; |\; \exists y. (x,y) \in
   R[R := R + (i,j)]^\mathcal{I} \wedge\\ y \in C[R := R +
   (i,j)]^\mathcal{I}\}$ that is $((\exists R. C)[R := R + (i,j)])^\mathcal{I} =\\ \{x\; |\; \exists y. (x,y) \in
   R^\mathcal{I}\cup\{(i,j)\} \wedge y \in C[R := R +
   (i,j)]^\mathcal{I}\} =\\ \{x\; |\; (\exists y. (x,y) \in
   R^\mathcal{I} \wedge y \in C[R := R +
   (i,j)]^\mathcal{I}) \vee (x = i \wedge\\ j \in C[R := R + (i,j)])\} =
   \{x\; |\; (x = i \Rightarrow \exists y. (x,y) \in
   R^\mathcal{I} \wedge \\y \in C[R := R +
   (i,j)]^\mathcal{I} \vee j \in C[R := R + (i,j)]) \wedge (x \neq i
   \Rightarrow \exists y. (x,y) \in
   R^\mathcal{I} \wedge\\ y \in C[R := R +
   (i,j)]^\mathcal{I})\}$. Moreover, $((o_i \Rightarrow \exists R. (C[R
   := R + (i,j)]) \sqcup\\ \exists U.(o_j \sqcap C[R := R + (i,j)]))
   \sqcap (\neg o_i \Rightarrow \exists R. (C[R := R + (i,j)])))^\mathcal{I} =\\ \{x\; |\; (x = i \Rightarrow \exists y. (x,y) \in
   R^\mathcal{I} \wedge y \in C[R := R +
   (i,j)]^\mathcal{I} \vee \\j \in C[R := R + (i,j)]) \wedge (x \neq i
   \Rightarrow \exists y. (x,y) \in
   R^\mathcal{I} \wedge y \in C[R := R +
   (i,j)]^\mathcal{I})\}$. Thus
 $((\exists R .C)[R := R + (i,j)])^\mathcal{I} =\\
   ((o_i \Rightarrow \exists R. (C[R
   := R + (i,j)]) \sqcup \exists U.(o_j \sqcap C[R := R + (i,j)]))
   \sqcap\\ (\neg o_i \Rightarrow \exists R. (C[R := R + (i,j)])))^\mathcal{I}$
\item One can see that $(\exists S^-. C)[S := S + (i,j)]$ is
  similar to\\ $(\exists S^-. C)[S^- := S^- + (j,i)]$, that is
  replacing $S$ by $S^-$ and swapping $i$ and $j$ in the previous case.
 \item As $(\exists R. C)^\mathcal{I} = \{x\; |\; \exists y. (x,y) \in
   R^\mathcal{I} \wedge y \in C^\mathcal{I}\}$,\\
 $((\exists R .C)[R := R - (i,j)])^\mathcal{I} = \{x\; |\; \exists y. (x,y) \in
   R[R := R - (i,j)]^\mathcal{I} \wedge\\ y \in C[R := R -
   (i,j)]^\mathcal{I}\}$ that is $((\exists R. C)[R := R - (i,j)])^\mathcal{I} =\\ \{x\; |\; \exists y. (x,y) \in
   R^\mathcal{I}\cap \overline{\{(i,j)\}} \wedge y \in C[R := R -
   (i,j)]^\mathcal{I}\} =\\ \{x\; |\; \exists y. (x \neq i \vee y \neq j)
   \wedge ((x,y) \in
   R^\mathcal{I} \wedge y \in C[R := R -
   (i,j)]^\mathcal{I})\} = \{x\; |\; (x = i \Rightarrow \exists y. (x,y) \in
   R^\mathcal{I} \wedge y \in C[R := R -
   (i,j)]^\mathcal{I} \wedge y \neq j) \wedge\\ (x \neq i \Rightarrow \exists y. (x,y) \in
   R^\mathcal{I} \wedge y \in C[R := R -
   (i,j)]^\mathcal{I})\}$. Moreover,\\ $((o_i \Rightarrow \exists R. (C[R
   := R -(i,j)] \sqcup \neg o_j)) \sqcap (\neg o_i
   \Rightarrow \exists R .(C[R := R -
   (i,j)])))^\mathcal{I} = \{x\; |\; (x = i \Rightarrow \exists y. (x,y) \in
   R^\mathcal{I} \wedge y \in C[R := R -
   (i,j)]^\mathcal{I} \wedge y \neq j) \wedge\\ (x \neq i \Rightarrow \exists y. (x,y) \in
   R^\mathcal{I} \wedge y \in C[R := R -
   (i,j)]^\mathcal{I})\}$. Thus\\
 $((\exists R .C)[R := R - (i,j)])^\mathcal{I} = \\
  ((o_i \Rightarrow \exists R. (C[R
   := R -(i,j)] \sqcup \neg o_j)) \sqcap (\neg o_i
   \Rightarrow \exists R .(C[R := R -
   (i,j)])))^\mathcal{I}$
\item One can see that $(\exists S^-. C)[S := S - (i,j)]$ is
  similar to\\ $(\exists S^-. C)[S^- := S^- - (j,i)]$, that is
  replacing $S$ by $S^-$ and swapping $i$ and $j$ in the previous case.
 \item As $(\forall R. C)^\mathcal{I} = \{x\; |\; \forall y. (x,y) \in
   R^\mathcal{I} \Rightarrow y \in C^\mathcal{I}\}$, $((\forall R .C)[c
   := c \pm i])^\mathcal{I} = \{x\; |\; \forall y. (x,y) \in
   R^\mathcal{I} \Rightarrow y \in C[c := c \pm i]^\mathcal{I}\}$ as
   $R^\mathcal{I}$ is independent of the valuation of $c$. Thus,
   $((\forall R .C)[c := c \pm i])^\mathcal{I} = (\forall R .(C[c := c \pm
   i]))^\mathcal{I}$
 \item As $(\forall R^-. C)^\mathcal{I} = \{x\; |\; \forall y. (y,x) \in
   R^\mathcal{I} \Rightarrow y \in C^\mathcal{I}\}$, $((\forall R^-.C)[c
   := c \pm i])^\mathcal{I} = \{x\; |\; \forall y. (y,x) \in
   R^\mathcal{I} \Rightarrow y \in C[c := c \pm i]^\mathcal{I}\}$ as
   $R^\mathcal{I}$ is independent of the valuation of $c$. Thus,
   $((\forall R^-.C)[c := c \pm i])^\mathcal{I} = (\forall R^-.(C[c := c \pm
   i]))^\mathcal{I}$
 \item As $(\forall R. C)^\mathcal{I} = \{x\; |\; \forall y. (x,y) \in
   R^\mathcal{I} \Rightarrow y \in C^\mathcal{I}\}$, $((\forall R .C)[R'
   := R' \pm (i,j)])^\mathcal{I} = \{x\; |\; \forall y. (x,y) \in
   R^\mathcal{I} \Rightarrow y \in C[R' := R' \pm (i,j)]^\mathcal{I}\}$ as
   $R^\mathcal{I}$ is independent of the valuation of
   $R'$. Thus,$((\forall R .C)[R' := R' \pm (i,j)])^\mathcal{I} =
   (\forall R .(C[R' := R' \pm (i,j)]))^\mathcal{I}$
 \item As $(\forall R^-. C)^\mathcal{I} = \{x\; |\; \forall y. (y,x) \in
   R^\mathcal{I} \Rightarrow y \in C^\mathcal{I}\}$,\\ $((\forall R^-.C)[R'
   := R' \pm (i,j)])^\mathcal{I} = \{x\; |\; \forall y. (y,x) \in
   R^\mathcal{I} \Rightarrow\\ y \in C[R' := R' \pm (i,j)]^\mathcal{I}\}$ as
   $R^\mathcal{I}$ is independent of the valuation of $R'$. Thus,$((\forall R^-.C)[R' := R' \pm (i,j)])^\mathcal{I} = (\forall R^-.(C[R' := R' \pm (i,j)]))^\mathcal{I}$
 \item As $(\forall R. C)^\mathcal{I} = \{x\; |\; \forall y. (x,y) \in
   R^\mathcal{I} \Rightarrow y \in C^\mathcal{I}\}$,\\ $((\forall R .C)[R
   := R + (i,j)])^\mathcal{I} = \{x\; |\; \forall y. (x,y) \in
   R^\mathcal{I} \cup \{(i,j)\} \Rightarrow\\ y \in C[R := R +
   (i,j)]^\mathcal{I}\}$. That is $((\forall R .C)[R
   := R + (i,j)])^\mathcal{I} =\\ \{x\; |\; (\forall y. (x,y) \in
   R^\mathcal{I} \Rightarrow y \in C[R := R +
   (i,j)]^\mathcal{I}) \wedge\\ (x = i \Rightarrow  j \in C[R := R +
   (i,j)]^\mathcal{I})\} = \{x\; |\; (x = i \Rightarrow (\forall y. (x,y) \in
   R^\mathcal{I} \Rightarrow\\ y \in C[R := R +
   (i,j)]^\mathcal{I} \wedge j \in (C[R := R + (i,j)])^\mathcal{I})) \wedge\\ (x \neq i \Rightarrow \forall y. (x,y) \in
   R^\mathcal{I} \Rightarrow y \in C[R := R +
   (i,j)]^\mathcal{I})\}$. Moreover,\\ $((o_i \Rightarrow \forall R. (C [R := R + (i,j)]) \sqcap \exists R. (o_j
   \sqcap C[R := R + (i,j)])) \sqcap\\ (\neg o_i \Rightarrow \forall R. (C [R := R + (i,j)]))^\mathcal{I} = \{x\; |\; (x = i \Rightarrow (\forall y. (x,y) \in
   R^\mathcal{I} \Rightarrow\\ y \in C[R := R +
   (i,j)]^\mathcal{I} \wedge j \in (C[R := R + (i,j)])^\mathcal{I})) \wedge (x \neq i \Rightarrow\\ \forall y. (x,y) \in
   R^\mathcal{I} \Rightarrow y \in C[R := R +
   (i,j)]^\mathcal{I})\}$. Thus,\\ $((\forall R .C)[R := R + (i,j)])^\mathcal{I}= 
 ((o_i \Rightarrow \forall R. (C [R := R + (i,j)]) \sqcap\\ \exists R. (o_j
   \sqcap C[R := R + (i,j)])) \sqcap (\neg o_i \Rightarrow \forall
   R. (C [R := R + (i,j)]))^\mathcal{I}$ 
\item One can see that $(\forall S^-. C)[S := S + (i,j)]$ is
  similar to\\ $(\forall S^-. C)[S^- := S^- + (j,i)]$, that is
  replacing $S$ by $S^-$ and swapping $i$ and $j$ in the previous case.
 \item As $(\forall R. C)^\mathcal{I} = \{x\; |\; \forall y. (x,y) \in
   R^\mathcal{I} \Rightarrow y \in C^\mathcal{I}\}$,\\ $((\forall R .C)[R
   := R - (i,j)])^\mathcal{I} = \{x\; |\; \forall y. (x,y) \in
   R^\mathcal{I} \cap \overline{\{(i,j)\}} \Rightarrow\\ y \in C[R := R -
   (i,j)]^\mathcal{I}\}$. That is $((\forall R .C)[R
   := R - (i,j)])^\mathcal{I} =\\ \{x\; |\; (x \neq i \wedge \forall y. ((x,y) \in
   R^\mathcal{I} \Rightarrow y \in C[R := R -
   (i,j)]^\mathcal{I})) \vee\\ (x = i \wedge \forall y. ((x,y) \in
   R^\mathcal{I} \Rightarrow y \in C[R := R -
   (i,j)]^\mathcal{I} \vee y = j))\} =\\ \{x\; |\; \forall y. (x \neq i \wedge (x,y) \in
   R^\mathcal{I} \Rightarrow y \in C[R := R -
   (i,j)]^\mathcal{I})\} \cap\\ \{x\; |\; \forall y. (x = i \wedge (x,y) \in
   R^\mathcal{I} \Rightarrow y \in C[R := R -
   (i,j)]^\mathcal{I} \vee y = j)\}$. Moreover, $(o_i \Rightarrow \forall
   R. (o_j \sqcup C [R := R - (i,j)]))^\mathcal{I} = \{x\; |\; \forall y. (x = i \wedge\\ (x,y) \in
   R^\mathcal{I} \Rightarrow y \in C[R := R -
   (i,j)]^\mathcal{I} \vee y = j)\}$ and\\ $(\neg o_i \Rightarrow 
  \forall R. (C [R := R - (i,j)]))^\mathcal{I} = \{x\; |\; \forall y. (x \neq i \wedge (x,y) \in
   R^\mathcal{I} \Rightarrow\\ y \in C[R := R -
   (i,j)]^\mathcal{I})\}$. Thus, $((\forall R .C)[R := R - (i,j)])^\mathcal{I} =\\
 ((o_i \Rightarrow \forall R. (C [R := R - (i,j)] \sqcup o_j)) \sqcap
 (\neg o_i \Rightarrow \forall R. (C[R := R - (i,j)])))^\mathcal{I}$
\item One can see that $(\forall S^-. C)[S := S - (i,j)]$ is
  similar to\\ $(\forall S^-. C)[S^- := S^- - (j,i)]$, that is
  replacing $S$ by $S^-$ and swapping $i$ and $j$ in the previous
  case.
\end{enumerate}
\end{proof}

Next we prove the \lemmaref{lemma2}.
\begin{framed}
Let $\Sigma$ be a signature, $\mathcal{T}$ is terminating.
\end{framed}
\begin{proof}
To prove the termination, we introduce a pre-ordering
($\mathcal{M}, \mathcal{M}'$) defined
as:
\begin{itemize}
\item $\mathcal{M}(\bot) = 0$
\item $\mathcal{M}(c) = 0$
\item $\mathcal{M}(\neg\; C) = \mathcal{M}(C)$
\item $\mathcal{M}(C\; \sqcap\; D) = max(\mathcal{M}(C),
  \mathcal{M}(D))$
\item $\mathcal{M}(C\; \sqcup\; D) = max(\mathcal{M}(C),
  \mathcal{M}(D))$
\item $\mathcal{M}((\geq n\; S\; C)) = \mathcal{M}((\geq n\; S^-\; C)) =
  \mathcal{M}(C) + 1$
\item $\mathcal{M}((< n\; S\; C)) = \mathcal{M}((< n\; S^-\; C)) =
  \mathcal{M}(C) + 1$
\item $\mathcal{M}(\exists R. C) = \mathcal{M}(\exists R^-. C) =
  \mathcal{M}(C) + 1$
\item $\mathcal{M}(\forall R. C) = \mathcal{M}(\forall R^-. C) =
  \mathcal{M}(C) + 1$
\item $\mathcal{M}(o) = 0$
\item $\mathcal{M}(\exists R.Self) = \mathcal{M}(\exists R^-.Self) = 0$
\item $\mathcal{M}(C \theta) = \mathcal{M}(C) +1$ 
\item $\mathcal{M}'(\bot) = 0$ 
\item $\mathcal{M}'(c) = 0$
\item $\mathcal{M}'((\neg\; C)) = \mathcal{M}'(C)$
\item $\mathcal{M}'(C\; \sqcap\; D) = max(\mathcal{M}'(C),
  \mathcal{M}'(D))$
\item $\mathcal{M}'(C\; \sqcup\; D) = max(\mathcal{M}'(C),
  \mathcal{M}'(D))$
\item $\mathcal{M}'((\geq n\; S\; C)) = \mathcal{M}'((\geq n\; S^-\; C))
  = \mathcal{M}'(C)$
\item $\mathcal{M}'((< n\; S\; C)) = \mathcal{M}'((< n\; S^-\; C)) = \mathcal{M}'(C)$
\item $\mathcal{M}'((\exists R. C)) = \mathcal{M}'((\exists R^-. C)) = \mathcal{M}'(C)$
\item $\mathcal{M}'((\forall R. C)) = \mathcal{M}'((\forall R^-. C)) = \mathcal{M}'(C)$
\item $\mathcal{M}'(o) = 0$
\item $\mathcal{M}'(\exists R.Self) = \mathcal{M}'(\exists R^-.Self) = 0$
\item $\mathcal{M}'(\bot \theta) = 0$ 
\item $\mathcal{M}'(c \theta) = 0$
\item $\mathcal{M}'((\neg\; C) \theta) = \mathcal{M}'(C \theta) + 1$
\item $\mathcal{M}'((C\; \sqcap\; D) \theta) = max(\mathcal{M}'(C \theta),
  \mathcal{M}'(D \theta)) + 1$
\item $\mathcal{M}'((C\; \sqcup\; D) \theta) = max(\mathcal{M}'(C \theta),
  \mathcal{M}'(D \theta)) + 1$
\item $\mathcal{M}'((\geq n\; S\; C) \theta) = \mathcal{M}'((\geq n\;
  S^-\; C) \theta) = \mathcal{M}'(C \theta)
  + 1$
\item $\mathcal{M}'((< n\; S\; C) \theta) = \mathcal{M}'((< n\; S^-\; C)
  \theta) = \mathcal{M}'(C \theta) +
  1$
\item $\mathcal{M}'((\exists R. C) \theta) = \mathcal{M}'((\exists
  ^-R. C) \theta) = \mathcal{M}'(C \theta) +1$
\item $\mathcal{M}'((\forall R. C) \theta) = \mathcal{M}'((\forall
  R^-. C) \theta) = \mathcal{M}'(C \theta) +1$
\item $\mathcal{M}'(o \theta) = 0$
\item $\mathcal{M}'((\exists R.Self) \theta) = \mathcal{M}'((\exists
  R^-.Self) \theta) = 0$
\end{itemize}
For every concept $C$, $\mathcal{M}(C)$ and $\mathcal{M}'(C)$ are positive. We now
prove that the transformations either strictly decrease $\mathcal{M}$
or keep
$\mathcal{M}$ constant and strictly decrease $\mathcal{M}'$. We
compute the results of the functions for the left- and the right- hand
side for each transformation.
\begin{enumerate}
\item \mbox{}\\
$\begin{array}{ccccc}
\mathcal{M}(\bot\; \theta)& = & \mathcal{M}(\bot) +1& = & 1\\
\mathcal{M}(\bot)& & &= & 0  
\end{array}$
\item \mbox{}\\
$\begin{array}{ccccc}
\mathcal{M}(o\; \theta)& =& \mathcal{M}(o) + 1& =& 1\\
\mathcal{M}(o)& & &=& 0
\end{array}$
\item \mbox{}\\
$\begin{array}{ccccc}
\mathcal{M}(c [R := R \pm (i,j)])& =&  \mathcal{M}(c) + 1& =& 1\\
\mathcal{M}(c)& & &= & 0
\end{array}$
\item \mbox{}\\
$\begin{array}{ccccc}
\mathcal{M}(c [c' := c' \pm i])& =& \mathcal{M}(c) + 1& =& 1 \\
\mathcal{M}(c)& & & =& 0
\end{array}$
\item \mbox{}\\
$\begin{array}{ccccc}
\mathcal{M}(c [c := c + i])& =& \mathcal{M}(c) + 1& =& 1 \\
\mathcal{M}(c \sqcup o_i)& =& max(\mathcal{M}(c), \mathcal{M}(o_i))&=& 0
\end{array}$
\item \mbox{}\\
$\begin{array}{ccccc}
\mathcal{M}(c [c := c - i])& =& \mathcal{M}(c) + 1& =& 1 \\
\mathcal{M}(c \sqcap \neg o_i)& =& max(\mathcal{M}(c),
  \mathcal{M}(\neg o_i))& =& 0
\end{array}$
\item \mbox{}\\
$\begin{array}{ccccc}
\mathcal{M}((\neg C)\; \theta)& =& \mathcal{M}(\neg C) + 1& =&
  \mathcal{M}(C) + 1\\
\mathcal{M}(\neg(C\; \theta))& =& \mathcal{M}(C\; \theta)& =& \mathcal{M}(C) + 1 \\
\mathcal{M}'((\neg C)\; \theta)&
  =& \mathcal{M}'(C \theta) + 1 \\
 \mathcal{M}'(\neg (C \theta))& =&
  \mathcal{M}'(C \theta)
\end{array}$
\item \mbox{}\\
$\begin{array}{ccc}
\mathcal{M}((C \sqcup D)\;\theta)& =& max(\mathcal{M}(C),
  \mathcal{M}(D)) + 1\\
 \mathcal{M}(C \theta \sqcup D \theta) & =& max(\mathcal{M}(C) + 1,
  \mathcal{M}(D) + 1)\\ 
\mathcal{M}'((C \sqcup D)\;\theta)& =& max(\mathcal{M}'(C \theta),
  \mathcal{M}'(D \theta)) + 1 \\
\mathcal{M}'(C \theta \sqcup D \theta) & = &max(\mathcal{M}'(C \theta),
  \mathcal{M}'(D \theta))
\end{array}$
\item  \mbox{}\\
$\begin{array}{ccc}
\mathcal{M}((C \sqcap D)\;\theta)& =& max(\mathcal{M}(C),
  \mathcal{M}(D)) + 1 \\
\mathcal{M}(C \theta \sqcap D \theta)& =& max(\mathcal{M}(C) + 1,
  \mathcal{M}(D) + 1) \\
\mathcal{M}'((C \sqcap D)\;\theta)& =& max(\mathcal{M}'(C \theta),
  \mathcal{M}'(D \theta)) + 1 \\
 \mathcal{M}'(C \theta \sqcap D \theta) & =& max(\mathcal{M}'(C \theta),
  \mathcal{M}'(D \theta))
\end{array}$
\item \mbox{}\\
$\begin{array}{ccccc}
\mathcal{M}(\exists R.Self [c := c \pm i])& =&
  \mathcal{M}(\exists R.Self) + 1& =& 1 \\
\mathcal{M}(\exists
  R.Self)& &  &=& 0
\end{array}$
\item As the definitions of $\mathcal{M}$ and $\mathcal{M}'$ do not
  discriminate $R$ and $R^-$, this rule is identic to the previous one.
\item \mbox{}\\
$\begin{array}{ccccc}
\mathcal{M}(\exists R.Self [R' := R' \pm (i,j)]) &=&
  \mathcal{M}(\exists R.Self) + 1 &=& 1 \\
\mathcal{M}(\exists
  R.Self)& & &=& 0
\end{array}$
\item As the definitions of $\mathcal{M}$ and $\mathcal{M}'$ do not
  discriminate $R$ and $R^-$, this rule is identic to the previous one.
\item \mbox{}\\
$\begin{array}{ccccc}
\mathcal{M}(\exists R.Self [R := R + (i,j)]) &=&
  \mathcal{M}(\exists R.Self) + 1 &= &1 \\
\mathcal{M}( (o_i \sqcap
  o_j) \sqcup \exists R.Self) &=& 0
\end{array}$
\item As the definitions of $\mathcal{M}$ and $\mathcal{M}'$ do not
  discriminate $R$ and $R^-$, this rule is identic to the previous one.
\item \mbox{}\\
$\begin{array}{ccccc}
\mathcal{M}(\exists R.Self [R := R - (i,j)]) &=&
  \mathcal{M}(\exists R.Self) + 1 &=& 1 \\
\mathcal{M}((\neg o_i \sqcup
 \neg o_j) \sqcap \exists R.Self)& & &=& 0
\end{array}$
\item As the definitions of $\mathcal{M}$ and $\mathcal{M}'$ do not
  discriminate $R$ and $R^-$, this rule is identic to the previous one.
\item \mbox{}\\
$\begin{array}{ccccc}
\mathcal{M}((\bowtie\; n\; S\; C) [c := c \pm i]) &=&
  \mathcal{M}((\bowtie\; n\; S\; C)) + 1 &=& \mathcal{M}(C) + 2\\
\mathcal{M}((\bowtie\; n\; S\; C [c := c \pm i]))& =& \mathcal{M}(C [c
  := c \pm i]) + 1 &=& \mathcal{M}(C) + 2 \\
\mathcal{M}'((\bowtie\; n\; S\; C) [c := c \pm i]) &= &\mathcal{M}'(C
  [c := c \pm i]) + 1 \\
\mathcal{M}'((\bowtie\; n\; S\; C [c := c \pm i]))& =& \mathcal{M}'(C
  [c := c \pm i])
\end{array}$
\item As the definitions of $\mathcal{M}$ and $\mathcal{M}'$ do not
  discriminate $R$ and $R^-$, this rule is identic to the previous one.
\item \mbox{}\\
$\begin{array}{ccccc}
\mathcal{M}((\bowtie\; n\; S\; C) [R' := R' \pm (i,j)]) &=&
  \mathcal{M}((\bowtie\; n\; S\; C)) + 1 &=& \mathcal{M}(C) + 2 \\
\mathcal{M}( (\bowtie\; n\;
  S\; C [R' := R' \pm (i,j)]) & =& \mathcal{M}(C [R' := R' \pm (i,j)]) + 1
  & =& \mathcal{M}(C) + 2 \\
\mathcal{M}'((\bowtie\; n\; S\; C)
  [R' := R' \pm (i,j)]) &=& \mathcal{M}'(C [R' := R' \pm (i,j)]) + 1 \\
\mathcal{M}'( (\bowtie\; n\;
  S\; C [R' := R' \pm (i,j)]))& =& \mathcal{M}'(C [R' := R' \pm (i,j)])
\end{array}$
\item As the definitions of $\mathcal{M}$ and $\mathcal{M}'$ do not
  discriminate $R$ and $R^-$, this rule is identic to the previous one.
\item \mbox{}\\
$\begin{array}{ccc}
\mathcal{M}((\bowtie\; n\; S\; C) [S := S + (i,j)]) &=&
  \mathcal{M}((\bowtie\; n\; S\; C) + 1 \\
&=& \mathcal{M}(C) + 2 \\
\mathcal{M}(c_{1,16}) &=& max(\mathcal{M}(o_j), \mathcal{M}(C[S := S +
(i,j)])) \\
&=& \mathcal{M}(C) + 1\\
\mathcal{M}(c_{2,16}) &=& max(\mathcal{M}(\neg o_j), \mathcal{M}(\neg (C[S := S +
(i,j)]))) \\
&=& \mathcal{M}(C) + 1\\
\mathcal{M}(b_{1,16}) &=& \mathcal{M}(o_i)\\
 &=& 0\\
\mathcal{M}(b_{2,16}) &=& \mathcal{M}(\exists U. c_{1,16}) \\
& =&
\mathcal{M}(c_{1,16}) + 1 \\
&=& \mathcal{M}(C)
+ 1\\
\mathcal{M}(b_{3,16}) &=& \mathcal{M}(\forall S. \neg o_j)\\
&=&
\mathcal{M}(\neg o_J) + 1 \\
& = & 1\\
\mathcal{M}(b_{4,16}) &=& \mathcal{M}((\bowtie\; (n-1)\; S\; C[S := S
+ (i,j)]) )\\
&=&
\mathcal{M}(C[S := S + (i,j)]) + 1 \\
& = & \mathcal{M}(C) + 2\\
\mathcal{M}(b_{11,16}) &=& \mathcal{M}(\neg o_i) \\
&=& 0\\
\mathcal{M}(b_{12,16}) &=& \mathcal{M}(\forall U. c_{2,16}) \\
& =&
\mathcal{M}(c_{2,16}) + 1\\
 &=& \mathcal{M}(C)
+ 2\\
\mathcal{M}(b_{13,16}) &=& \mathcal{M}(\exists S. o_j)\\
&=&
\mathcal{M}(o_J) + 1 \\
& = & 1\\
\end{array}$
\newpage
$\begin{array}{ccc}
\mathcal{M}(b_{14,16}) &=& \mathcal{M}((\bowtie\; n\; S\; C[S := S
+ (i,j)]) ) \\
&=&
\mathcal{M}(C[S := S + (i,j)]) + 1 \\
& = & \mathcal{M}(C) + 2\\
\mathcal{M}(a_{1,16}) &=& max(\mathcal{M}(b_{1,16}), 
\mathcal{M}(b_{2,16}), \mathcal{M}(b_{3,16}), \mathcal{M}(b_{4,16}) \\
& = &
  \mathcal{M}(C) + 2\\
\mathcal{M}(a_{2,16}) &=& max(\mathcal{M}(b_{11,16}),
\mathcal{M}(b_{12,16}), \mathcal{M}(b_{13,16}), \mathcal{M}(b_{14,16})
\\
& = & \mathcal{M}(C) + 2\\
\mathcal{M}(RHS_{16}) &=&  max(\mathcal{M}(a_{1,16}),
\mathcal{M}(a_{2,16})) \\
&=& \mathcal{M}(C) + 2\\
\mathcal{M}'((\bowtie\; n\; S\; C) [S := S + (i,j)]) &=&
\mathcal{M}'(C [S := S + (i,j)]) + 1 \\ 
\mathcal{M}'(c_{1,16}) &=& max(\mathcal{M}'(o_j), \mathcal{M}'(C[S := S +
(i,j)])) \\
&=& \mathcal{M}'(C [S := S + (i,j)])\\
\mathcal{M}'(c_{2,16}) &=& max(\mathcal{M}'(\neg o_j), \mathcal{M}'(\neg (C[S := S +
(i,j)]))) \\
&=& \mathcal{M}'(C [S := S + (i,j)])\\
\mathcal{M}'(b_{1,16}) &=& \mathcal{M}(o_i) \\
&=& 0\\
\mathcal{M}'(b_{2,16}) &=& \mathcal{M}(\exists U. c_{1,16}) \\
& =&
\mathcal{M}'(c_{1,16})\\
&=& \mathcal{M}'(C [S := S + (i,j)])\\
\mathcal{M}'(b_{3,16}) &=& \mathcal{M}(\forall S. \neg o_j)\\
&=&
\mathcal{M}'(\neg o_J)\\
& = & 0\\
\mathcal{M}'(b_{4,16}) &=& \mathcal{M}((\bowtie\; (n-1)\; S\; C[S := S
+ (i,j)]) \\
&=&
\mathcal{M}'(C[S := S + (i,j)])\\
\mathcal{M}'(b_{11,16}) &=& \mathcal{M}'(\neg o_i) \\
&=& 0\\
\mathcal{M}'(b_{12,16}) &=& \mathcal{M}'(\forall U. c_{2,16}) \\
& =&
\mathcal{M}'(c_{2,16})\\
&=& \mathcal{M}(C [S := S + (i,j])\\
\mathcal{M}'(b_{13,16}) &=& \mathcal{M}'(\exists S. o_j)\\
&=&
\mathcal{M}'(o_J)\\
& = & 0\\
\mathcal{M}'(b_{14,16}) &=& \mathcal{M}'((\bowtie\; n\; S\; C[S := S
+ (i,j)]) )\\
&=&
\mathcal{M}'(C[S := S + (i,j)])\\
\mathcal{M}'(a_{1,16}) &=& max(\mathcal{M}'(b_{1,16}), 
\mathcal{M}'(b_{2,16}), \mathcal{M}'(b_{3,16}),
\mathcal{M}'(b_{4,16})\\
 & = &
  \mathcal{M}'(C [S := S +(i,j)])\\
\mathcal{M}'(a_{2,16}) &=& max(\mathcal{M}'(b_{11,16}),
\mathcal{M}'(b_{12,16}), \mathcal{M}'(b_{13,16}),
\mathcal{M}'(b_{14,16}) \\
& = & \mathcal{M}'(C [S := S + (i,j)])\\
\mathcal{M}'(RHS_{16}) &=&  max(\mathcal{M}'(a_{1,16}),
\mathcal{M}'(a_{2,16})) \\
&=& \mathcal{M}(C [S := S + (i,j)])\\
\end{array}$
\item As the definitions of $\mathcal{M}$ and $\mathcal{M}'$ do not
  discriminate $R$ and $R^-$, this rule is identic to the previous one.
\item \mbox{}\\
$\begin{array}{ccc}
\mathcal{M}((\bowtie\; n\; S\; C) [S := S - (i,j)]) &=&
  \mathcal{M}((\bowtie\; n\; S\; C) + 1 \\
&=& \mathcal{M}(C) + 2 \\
\mathcal{M}(c_{1,17}) &=& max(\mathcal{M}(o_j), \mathcal{M}(C[S := S -
(i,j)])) \\
&=& \mathcal{M}(C) + 1\\
\mathcal{M}(c_{2,17}) &=& max(\mathcal{M}(\neg o_j), \mathcal{M}(\neg (C[S := S -
(i,j)]))) \\
&=& \mathcal{M}(C) + 1\\
\mathcal{M}(b_{1,17}) &=& \mathcal{M}(o_i)\\
 &=& 0\\
\mathcal{M}(b_{2,17}) &=& \mathcal{M}(\exists U. c_{1,17}) \\
& =&
\mathcal{M}(c_{1,17}) + 1\\
&=& \mathcal{M}(C)
+ 2\\
\mathcal{M}(b_{3,17}) &=& \mathcal{M}(\exists S. o_j)\\
&=&
\mathcal{M}(\neg o_J) + 1 \\
& = & 1\\
\mathcal{M}(b_{4,17}) &=& \mathcal{M}((\bowtie\; (n+1)\; S\; C[S := S
- (i,j)]) )\\
&=&
\mathcal{M}(C[S := S - (i,j)]) + 1 \\
& = & \mathcal{M}(C) + 2\\
\mathcal{M}(b_{11,17}) &=& \mathcal{M}(\neg o_i) \\
&=& 0\\
\mathcal{M}(b_{12,17}) &=& \mathcal{M}(\forall U. c_{2,17}) \\
& =&
\mathcal{M}(c_{2,17}) + 1\\
 &=& \mathcal{M}(C)
+ 2\\
\mathcal{M}(b_{13,17}) &=& \mathcal{M}(\forall S.\neg o_j)\\
&=&
\mathcal{M}(o_J) + 1 \\
& = & 1\\
\mathcal{M}(b_{14,17}) &=& \mathcal{M}((\bowtie\; n\; S\; C[S := S
+ (i,j)]) ) \\
&=&
\mathcal{M}(C[S := S - (i,j)]) + 1 \\
& = & \mathcal{M}(C) + 2\\
\mathcal{M}(a_{1,17}) &=& max(\mathcal{M}(b_{1,17}), 
\mathcal{M}(b_{2,17}), \mathcal{M}(b_{3,17}), \mathcal{M}(b_{4,17}) \\
& = &
  \mathcal{M}(C) + 2\\
\mathcal{M}(a_{2,17}) &=& max(\mathcal{M}(b_{11,17}),
\mathcal{M}(b_{12,17}), \mathcal{M}(b_{13,17}), \mathcal{M}(b_{14,17})
\\
& = & \mathcal{M}(C) + 2\\
\mathcal{M}(RHS_{17}) &=&  max(\mathcal{M}(a_{1,17}),
\mathcal{M}(a_{2,17})) \\
&=& \mathcal{M}(C) + 2\\
\mathcal{M}'((\bowtie\; n\; S\; C) [S := S - (i,j)]) &=&
\mathcal{M}'(C [S := S - (i,j)]) + 1 \\ 
\mathcal{M}'(c_{1,17}) &=& max(\mathcal{M}'(o_j), \mathcal{M}'(C[S := S -
(i,j)])) \\
&=& \mathcal{M}'(C [S := S - (i,j)])\\
\mathcal{M}'(c_{2,17}) &=& max(\mathcal{M}'(\neg o_j), \mathcal{M}'(\neg (C[S := S -
(i,j)]))) \\
&=& \mathcal{M}'(C [S := S - (i,j)])\\
\mathcal{M}'(b_{1,17}) &=& \mathcal{M}(o_i) \\
&=& 0\\
\mathcal{M}'(b_{2,17}) &=& \mathcal{M}(\exists U. c_{1,17}) \\
& =&
\mathcal{M}'(c_{1,17})\\
&=& \mathcal{M}'(C [S := S - (i,j)])
\end{array}$
\newpage
$\begin{array}{ccc}
\mathcal{M}'(b_{3,17}) &=& \mathcal{M}(\forall S. \neg o_j)\\
&=&
\mathcal{M}'(\neg o_J)\\
& = & 0\\
\mathcal{M}'(b_{4,17}) &=& \mathcal{M}((\bowtie\; (n+1)\; S\; C[S := S
- (i,j)]) \\
&=&
\mathcal{M}'(C[S := S - (i,j)])\\
\mathcal{M}'(b_{11,17}) &=& \mathcal{M}'(\neg o_i) \\
&=& 0\\
\mathcal{M}'(b_{12,17}) &=& \mathcal{M}'(\forall U. c_{2,17}) \\
& =&
\mathcal{M}'(c_{2,17})\\
&=& \mathcal{M}(C [S := S - (i,j])\\
\mathcal{M}'(b_{13,17}) &=& \mathcal{M}'(\forall S. \neg o_j)\\
&=&
\mathcal{M}'(o_J)\\
& = & 0\\
\mathcal{M}'(b_{14,17}) &=& \mathcal{M}'((\bowtie\; n\; S\; C[S := S
- (i,j)]) )\\
&=&
\mathcal{M}'(C[S := S - (i,j)])\\
\mathcal{M}'(a_{1,17}) &=& max(\mathcal{M}'(b_{1,17}), 
\mathcal{M}'(b_{2,17}), \mathcal{M}'(b_{3,17}),
\mathcal{M}'(b_{4,17})\\
 & = &
  \mathcal{M}'(C [S := S - (i,j)])\\
\mathcal{M}'(a_{2,17}) &=& max(\mathcal{M}'(b_{11,17}),
\mathcal{M}'(b_{12,17}), \mathcal{M}'(b_{13,17}),
\mathcal{M}'(b_{14,17}) \\
& = & \mathcal{M}'(C [S := S - (i,j)])\\
\mathcal{M}'(RHS_{17}) &=&  max(\mathcal{M}'(a_{1,17}),
\mathcal{M}'(a_{2,17})) \\
&=& \mathcal{M}(C [S := S - (i,j)])\\
\end{array}$
\item As the definitions of $\mathcal{M}$ and $\mathcal{M}'$ do not
  discriminate $R$ and $R^-$, this rule is identic to the previous one.
\item \mbox{} \\
$\begin{array}{ccccc}
\mathcal{M}((\exists R .C)[c := c \pm i])& = &
\mathcal{M}(\exists R .C) + 1& = & \mathcal{M}(C) + 2\\
\mathcal{M}(\exists R .(C[c := c \pm
  i]))& = &\mathcal{M}(C[c := c \pm
  i]) + 1 &= & \mathcal{M}(C) + 2\\
\mathcal{M}'((\exists R .C)[c := c \pm i])& &
& =& \mathcal{M}'(C[c := c \pm i]) + 1\\
\mathcal{M}'(\exists R .(C[c := c \pm
  i]))& & &= & \mathcal{M}(C[c := c \pm i])\\
\end{array}$
\item As the definitions of $\mathcal{M}$ and $\mathcal{M}'$ do not
  discriminate $R$ and $R^-$, this rule is identic to the previous one.
\item \mbox{}\\
$\begin{array}{ccccc} 
\mathcal{M}((\exists R .C)[R' := R' \pm (i,j)])& = &
\mathcal{M}(\exists R .C) + 1& = & \mathcal{M}(C) + 2\\
\mathcal{M}(\exists R .(C[R' := R' \pm
  (i,j)]))& = &\mathcal{M}(C[R' := R' \pm
  (i,j)]) + 1 &= & \mathcal{M}(C) + 2\\
\mathcal{M}'((\exists R .C)[R' := R' \pm (i,j)])& &
& =& \mathcal{M}'(C[R' := R' \pm (i,j)]) + 1\\
\mathcal{M}'(\exists R .(C[R' := R' \pm
  (i,j)]))& & &= & \mathcal{M}(C[R' := R' \pm (i,j)])
\end{array}$
\item As the definitions of $\mathcal{M}$ and $\mathcal{M}'$ do not
  discriminate $R$ and $R^-$, this rule is identic to the previous one.
\item \mbox{}\\
$\begin{array}{cclll}  
\mathcal{M}((\exists R .C)[R := R + (i,j)]) & = & \mathcal{M}(\exists
R . C) + 1 \\
&=&  \mathcal{M}(C) + 2\\
\mathcal{M}(d_1) &= & max(\mathcal{M}(\neg
o_i),\mathcal{M}(d_{1,1}),\mathcal{M}(d_{1,2}), \mathcal{M}(o_i),\mathcal{M}(d_{1,3}))\\
&=& 
max(\mathcal{M}(o_i),\mathcal{M}(d_{1,1,1}) +
2,max(\mathcal{M}(o_j),\mathcal{M}(d_{1,2,1})) + 1,\\
& & 0,\mathcal{M}(d_{1,3,1}) + 1)\\
&=& max(0,\mathcal{M}(C) + 2,max(0,\mathcal{M}(C) +1) + 1,
\mathcal{M}(C) + 2)\\
&=& \mathcal{M}(C) + 2\\
\mathcal{M}'((\exists R .C)[R := R + (i,j)]) & = & \mathcal{M}'(C [R
:= R + (i,j)]) + 1 \\
\mathcal{M}'(d_1) &= & max(\mathcal{M}'(\neg o_i),\mathcal{M}'(d_{1,1}),\mathcal{M}'(d_{1,2}),\mathcal{M}'(o_i),\mathcal{M}'(d_{1,3}))\\
&=& 
max(\mathcal{M}'(o_i),\mathcal{M}'(C[R := R + (i,j)]), \\
& & max(\mathcal{M}'(o_j),\mathcal{M}'(C[R := R + (i,j)])),0,\mathcal{M}'(C[R := R + (i,j)]))\\
&=& \mathcal{M}'(C[R := R + (i,j)])\\
\end{array}$
\item As the definitions of $\mathcal{M}$ and $\mathcal{M}'$ do not
  discriminate $R$ and $R^-$, this rule is identic to the previous one.
\item \mbox{}\\
$\begin{array}{cclll}
\mathcal{M}((\exists R .C)[R := R - (i,j)]) &=& \mathcal{M}(\exists
R.C) + 1\\
&=& \mathcal{M}(C) + 2\\
\mathcal{M}(d_2) &=& max(\mathcal{M}(\neg o_i),\mathcal{M}(d_{2,1}), \mathcal{M}(o_i),\mathcal{M}(d_{2,2}))\\
&=& max(\mathcal{M}(o_i), max(\mathcal{M}(d_{2,1,1}),\mathcal{M}(\neg
o_j)) + 1,
0,\mathcal{M}(d_{2,2,1}) + 1)\\
&=& max(0, max(\mathcal{M}(C) + 1, \mathcal{M}(o_i)) + 1, \mathcal{M}(C) + 2)\\
&=& \mathcal{M}(C) + 2\\
\mathcal{M}'((\exists R .C)[R := R - (i,j)]) &=& \mathcal{M}'(C[R := R -
(i,j)]) + 1\\
\mathcal{M}'(d_2) &=& max(\mathcal{M}'(\neg o_i), \mathcal{M}'(d_{2,1}), \mathcal{M}'(o_i),
\mathcal{M}'(d_{2,2}))\\
&=& max(\mathcal{M}'(o_i), max(\mathcal{M}'(C [R := R - (i,j)]),
\mathcal{M}'(\neg o_j)),\\
& & 0,\mathcal{M}'(C[R := R -(i,j)]))\\
&=& max(0, max(\mathcal{M}'(C[R := R - (i,j)]), \mathcal{M}'(o_j)),\\
& & \mathcal{M}'(C[R := R
- (i,j)]))\\
&=& \mathcal{M}'(C[R := R - (i,j)])
\end{array}$
\item As the definitions of $\mathcal{M}$ and $\mathcal{M}'$ do not
  discriminate $R$ and $R^-$, this rule is identic to the previous one.
\item \mbox{}\\
$\begin{array}{ccccc} 
\mathcal{M}((\forall R .C)[c := c \pm i]) &=& \mathcal{M}(\forall
R. C) + 1\\
&=& \mathcal{M}(C) + 2\\
\mathcal{M}(\forall R .(C[c := c \pm i])) &=& \mathcal{M}(C[c := c \pm
i]) + 1\\
&=& \mathcal{M}(C) + 2\\
\mathcal{M}'((\forall R .C)[c := c \pm i]) &=& \mathcal{M}'(C[c :=
c\pm i]) + 1\\
\mathcal{M}'(\forall R .(C[c := c \pm i])) &=& \mathcal{M}'(C[c := c \pm
i])
\end{array}$
\item As the definitions of $\mathcal{M}$ and $\mathcal{M}'$ do not
  discriminate $R$ and $R^-$, this rule is identic to the previous one.
\item \mbox{}\\
$\begin{array}{ccccc} 
\mathcal{M}((\forall R .C)[R' := R' \pm (i,j)]]) &=& \mathcal{M}(\forall
R. C) + 1\\
&=& \mathcal{M}(C) + 2\\
\mathcal{M}(\forall R .(C[R' := R' \pm (i,j)])) &=& \mathcal{M}(C[R' := R' \pm
(i,j)]) + 1\\
&=& \mathcal{M}(C) + 2\\
\mathcal{M}'((\forall R .C)[R' := R' \pm (i,j)]) &=& \mathcal{M}'(C[R' :=
R' \pm (i,j)]) + 1\\
\mathcal{M}'(\forall R .(C[R' := R' \pm (i,j)])) &=& \mathcal{M}'(C[R' := R' \pm
(i,j)])\\
\end{array}$
\item As the definitions of $\mathcal{M}$ and $\mathcal{M}'$ do not
  discriminate $R$ and $R^-$, this rule is identic to the previous one.
\item \mbox{}\\
$\begin{array}{cclll}
\mathcal{M}((\forall R .C)[R := R + (i,j)]) &=& \mathcal{M}(\forall R.
C) + 1\\
&=& \mathcal{M}(C) + 2\\
\mathcal{M}(d_3) &=& max(\mathcal{M}(\neg o_i),\mathcal{M}(d_{3,1}), \mathcal{M}(d_{3,2}),
\mathcal{M}(o_i), \mathcal{M}(d_{3,3}))\\
&=& max(\mathcal{M}(o_i),
\mathcal{M}(d_{3,1,1}),max(\mathcal{M}(o_j),\mathcal{M}(d_{3,2,1})) + 1,\\
& & 0,\mathcal{M}(d_{3,3,1}) + 1)\\
&=& max(0, \mathcal{M}(C) + 2, \mathcal{M}(C) + 2)\\
&=& \mathcal{M}(C) + 2\\
\mathcal{M}'((\forall R .C)[R := R + (i,j)]) &=& \mathcal{M}(\forall R.
C[R := R + (i,j)]) + 1\\
\mathcal{M}'(d_3) &=& max(\mathcal{M}'(\neg o_i),\mathcal{M}'(d_{3,1}), \mathcal{M}'(d_{3,2}),
\mathcal{M}'(o_i), \mathcal{M}'(d_{3,3}))\\
&=& max(\mathcal{M}'(o_i),
\mathcal{M}'(C[R := R + (i,j)]),\\
& & max(\mathcal{M}'(o_j),\mathcal{M}(C[R := R + (i,j)])),
0,\mathcal{M}(C[R := R + (i,j)]))\\
&=& max(0, \mathcal{M}(C[R := R + (i,j)]), \mathcal{M}(C[R := R + (i,j)]))\\
&=& \mathcal{M}(C)\\
\end{array}$ 
\item As the definitions of $\mathcal{M}$ and $\mathcal{M}'$ do not
  discriminate $R$ and $R^-$, this rule is identic to the previous one.
\item \mbox{}\\
$\begin{array}{cclll}
\mathcal{M}((\forall R.C)[R := R - (i,j)]) &=& \mathcal{M}(\forall
R.C) + 1\\
&=& \mathcal{M}(C) + 2\\
\mathcal{M}(d_4) &=& max(\mathcal{M}(\neg o_i), \mathcal{M}(d_{4,1}),
\mathcal{M}(o_i), \mathcal{M}(d_{4,2}))\\
&=& max(\mathcal{M}(o_i), max(\mathcal{M}(d_{4,1,1}),
\mathcal{M}(o_j)) + 1,
0, \mathcal{M}(d_{4,2,1}) + 1)\\
&=& max(0, max(\mathcal{M}(C) + 1, 0) + 1, \mathcal{M}(C) + 2)\\
&=& \mathcal{M}(C) + 2\\
\mathcal{M}'((\forall R.C)[R := R - (i,j)]) &=& \mathcal{M}(C[R := R -
(i,j)]) + 1\\
\mathcal{M}'(d_4) &=& max(\mathcal{M}'(\neg o_i), \mathcal{M}'(d_{4,1}),
\mathcal{M}'(o_i), \mathcal{M}'(d_{4,2}))\\
&=& max(\mathcal{M}'(o_i), max(\mathcal{M}'(C[R := R -
(i,j)]),
\mathcal{M}'(o_j)),\\
& & 0, \mathcal{M}'(C[R := R -
(i,j)]))\\
&=& max(0, max(\mathcal{M}'(C[R := R -
(i,j)]), 0), \mathcal{M}'(C[R := R -
(i,j)]))\\
&=& \mathcal{M}'(C[R := R -
(i,j)])
\end{array}$
\item As the definitions of $\mathcal{M}$ and $\mathcal{M}'$ do not
  discriminate $R$ and $R^-$, this rule is identic to the previous one.
\end{enumerate}
where :\\
$\begin{array}{cccccc}
RHS_{16} &=& ((o_i \sqcap \exists U. (o_j\; \sqcap C [R := R + (i,j)])
  \sqcap \forall R. \neg o_j) &\Rightarrow& 
 (\bowtie\; (n-1)\; R\; C [R := R + (i,j)]) )\\
 &\sqcap& ((\neg o_i \sqcup \forall U. (\neg o_j
  \sqcup \neg C [R := R + (i,j)]) \sqcup \exists R. o_j)
  &\Rightarrow&
 (\bowtie\; n\; R\; C [R := R + (i,j)]))\\ 
RHS_{17} &=& ((o_i \sqcap \exists U. (o_j\; \sqcap C [R := R - (i,j)])
  \sqcap \exists R. o_j) &\Rightarrow&
 (\bowtie\; (n+1)\; R\; C [R := R - (i,j)]) )\\
 &\sqcap& ((\neg o_i \sqcup \forall U. (\neg o_j
  \sqcup \neg C [R := R - (i,j)]) \sqcup \forall R. \neg o_j)
  &\Rightarrow&
 (\bowtie\; n\; R\; C [R := R - (i,j)]))\\
a_{1,16} &=& (o_i \sqcap \exists U. (o_j\; \sqcap C [R := R + (i,j)])
  \sqcap \forall R. \neg o_j) &\Rightarrow&
 (\bowtie\; (n-1)\; R\; C [R := R + (i,j)])\\
a_{2,16} &=& (\neg o_i \sqcup \forall U. (\neg o_j
  \sqcup \neg C [R := R + (i,j)]) \sqcup \exists R. o_j)
  &\Rightarrow&
 (\bowtie\; n\; R\; C [R := R + (i,j)])\\
a_{1,17} &=& (o_i \sqcap \exists U. (o_j\; \sqcap C [R := R - (i,j)])
  \sqcap \exists R. o_j) &\Rightarrow&
 (\bowtie\; (n+1)\; R\; C [R := R - (i,j)])\\
a_{2,17} &=& (\neg o_i \sqcup \forall U. (\neg o_j
  \sqcup \neg C [R := R - (i,j)]) \sqcup \forall R. \neg o_j)
  &\Rightarrow&
 (\bowtie\; n\; R\; C [R := R - (i,j)])\\
b_{1,16} &=& o_i\\
b_{2,16} &=& \exists U. (o_j\; \sqcap C [R := R + (i,j)])\\
b_{3,16} &=& \forall R. \neg o_j\\
b_{4,16} &=&  (\bowtie\; (n-1)\; R\; C [R := R + (i,j)]) )\\
b_{11,16} &=& \neg o_i\\
b_{12,16} &=& \forall U. (\neg o_j\; \sqcup \neg (C [R := R +
(i,j)]))\\
b_{13,16} &=& \exists R. o_j\\
b_{14,16} &=&  (\bowtie\; n\; R\; C [R := R + (i,j)]) )\\
b_{1,17} &=& o_i\\
b_{2,17} &=& \exists U. (o_j\; \sqcap C [R := R - (i,j)])\\
b_{3,17} &=& \exists R. o_j\\
b_{4,17} &=&  (\bowtie\; (n+1)\; R\; C [R := R - (i,j)]) )\\
b_{11,17} &=& \neg o_i\\
b_{12,17} &=& \forall U. (\neg o_j\; \sqcup \neg (C [R := R -
(i,j)]))\\
b_{13,17} &=& \forall R. \neg o_j\\
b_{14,17} &=&  (\bowtie\; n\; R\; C [R := R - (i,j)]) )\\
c_{1,16} &=& (o_j\; \sqcap C [R := R + (i,j)])\\
c_{2,16} &=& (\neg o_j\; \sqcup \neg (C [R := R + (i,j)]))\\
c_{1,17} &=& (o_j\; \sqcap C [R := R - (i,j)])\\
c_{2,17} &=& (\neg o_j\; \sqcup \neg (C [R := R - (i,j)]))\\
\end{array}$
$\begin{array}{cclll}
d_{1} &=& (o_i \Rightarrow \exists R. (C [R := R + (i,j)]) \sqcup
\exists U. (o_j \sqcap C[R := R + (i,j)])) \sqcap
(\neg o_i \Rightarrow \exists R .(C[R := R + (i,j)])) \\
&=& (\neg o_i \sqcup \exists R. (C [R := R + (i,j)]) \sqcup
\exists U. (o_j \sqcap C[R := R + (i,j)])) \sqcap
(o_i \sqcup \exists R .(C[R := R + (i,j)]))\\
d_{1,1} &=& \exists R. (C [R := R + (i,j)])\\
d_{1,2} &=& \exists U. (o_j \sqcap C[R := R + (i,j)])\\
d_{1,3} &=& \exists R .(C[R := R + (i,j)])\\
d_{1,2,1} &=& o_j\\
d_{1,2,2} &=& C [R := R + (i,j)]\\
d_{1,3,1} &=& C[R := R + (i,j)]\\
d_2 &=&  (o_i \Rightarrow (\exists R. (C [R := R - (i,j)] \sqcap \neg
o_j))) \sqcap
(\neg o_i \Rightarrow \exists R. (C [R := R - (i,j)]))\\
&=& (\neg o_i \sqcup (\exists R. (C [R := R - (i,j)] \sqcap \neg
o_j))) \sqcap
(o_i \sqcup \exists R. (C [R := R - (i,j)]))\\
d_{2,1} &=& \exists R. (C [R := R - (i,j)] \sqcap \neg o_j)\\ 
d_{2,2} &=& \exists R. (C[R := R - (i,j)])\\
d_{2,1,1} &=&  C [R := R - (i,j)]\\
d_{2,1,2} &=& \neg o_j\\
d_{2,2,1} &=& C [R := R - (i,j)]\\
\end{array}$
\newpage
$\begin{array}{cclll}
d_3 &=& (o_i \Rightarrow (\forall R. (C [R := R + (i,j)]) \sqcap
\exists U. (o_j \sqcap C[R := R + (i,j)])))
\sqcap
 (\neg o_i \Rightarrow \forall R. (C [R := R + (i,j)]))\\
&=& (\neg o_i \sqcup (\forall R. (C [R := R + (i,j)]) \sqcap
\exists U. (o_j \sqcap C[R := R + (i,j)])))
\sqcap
 (o_i \sqcup \forall R. (C [R := R + (i,j)]))\\
d_{3,1} &=& \forall R. (C [R := R + (i,j)])\\
d_{3,2} &=& \exists U. (o_j \sqcap C [R := R + (i,j)])\\
d_{3,3} &=& \forall R. (C [R := R + (i,j)])\\
d_{3,2,1} &=& C [R := R + (i,j)]\\
d_{3,3,1} &=& C [R := R + (i,j)]\\
d_4 &=& (o_i \Rightarrow \forall R.(C[R := R - (i,j)] \sqcup o_j))
\sqcap (\neg o_i \Rightarrow \forall R.(C[R := R - (i,j)]))\\
&=& (\neg o_i \sqcup \forall R.(C[R := R - (i,j)] \sqcup o_j))
\sqcap (o_i \sqcup \forall R.(C[R := R - (i,j)]))\\
d_{4,1} &=& \forall R.(C[R := R - (i,j)] \sqcup o_j)\\
d_{4,2} &=& \forall R.(C[R := R - (i,j)])\\
d_{4,1,1} &=& C[R := R - (i,j)]\\
d_{4,2,1} &=& C[R := R - (i,j)]
\end{array}$
\end{proof}

\end{document}